\documentclass[12pt]{article}
\usepackage{amsmath}
\usepackage{amssymb}
\usepackage{amsthm}
\usepackage{mathtools}
\usepackage{url}
\usepackage[usenames,svgnames]{xcolor}
\usepackage{enumitem}
\usepackage{comment}
\usepackage{mathdots}
\usepackage{graphicx}
\usepackage{float}
\usepackage{enumitem}
\usepackage[symbol]{footmisc}
\usepackage[pagebackref=false]{hyperref} 
\usepackage[all]{hypcap} 

\theoremstyle{plain}
\theoremstyle{definition}
\newtheorem{theorem}{Theorem}[section]
\newtheorem{lemma}[theorem]{Lemma}
\newtheorem{definition-theorem}[theorem]{Definition-Theorem}
\newtheorem{definition-proposition}[theorem]{Definition-Proposition}

\newtheorem{proposition}[theorem]{Proposition}
\newtheorem{corollary}[theorem]{Corollary}
\newtheorem{example}{Example}[section]
\newtheorem{examples}{Example}[subsection]

\newtheorem{remark}{Remark}[section]

\newtheorem{definition}{Definition}[section]

\numberwithin{equation}{section} 

\DeclareMathOperator{\End}{End}
\DeclareMathOperator{\Span}{span}
\DeclareMathOperator{\sgn}{sgn}

\DeclareMathOperator{\GL}{GL}
\DeclareMathOperator{\SL}{SL}
\DeclareMathOperator{\SO}{SO}

\def\ra{{\rightarrow}}
\def\lra{{\longrightarrow}}
\def\mt{{\mapsto}}

\def\det{\mathrm {det}}
\def\Det{\mathrm {Det}}
\def\mod{\mathrm {mod}}
\def\Pf{\mathrm {Pf}}
\def\Ca{\mathrm {Ca}}
\def\Pl{\grP\grl}
\def\rank{\mathrm {rank}}
\def\Spin{\mathrm {Spin}}
\def\End{\mathrm {End}}
\def\Image{\mathrm {Image}}

\def\span{\mathrm {span}}

\def\span{\mathrm {span}}

\def\Gr{\mathrm {Gr}}
\def\res{\mathop{\mathrm {res}}\limits}

\def\res{\mathop{\mathrm{res}}\limits}

\def\&{&{\hskip -20pt}}

\def\be{\begin{equation}}
\def\ee{\end{equation}}

\def\bea{\begin{eqnarray}}
\def\eea{\end{eqnarray}}

\def\bt{\begin{theorem}}
\def\et{\end{theorem}}

\def\bex{\begin{example}\small \rm}
\def\eex{\end{example}}

\def\bexs{\begin{examples}\small \rm}
\def\eexs{\end{examples}}

\def\ra{\rightarrow}
\def \ss {\subset}

\def\br{\begin{remark}\small \rm \em}
\def\er{\end{remark}}

\def\AA{{\mathcal A}}

\def\CC{{\mathcal C}}

\def\FF {{\mathcal F}}

\def\HH{{\mathcal H}}
\def\II {{\mathcal I}}

\def\NN{{\mathcal N}}
\def\OO{{\mathcal O}}

\def\QQ {{\mathcal Q}}

\def\SS{{\mathcal S}}

\def\Cb{\mathbf{C}}

\def\Ib{\mathbf{I}}

\def\Nb{\mathbf{N}}
\def\Ob{\mathbf{0}}
\def\Pb{\mathbf{P}}

\def\Zb{\mathbf{Z}}

\def\sb{\mathbf{s}}

\def\vb{\mathbf{v}}
\def\wb{\mathbf{w}}

\def\0b{\boldsymbol{0}}

\def\grH{\mathfrak{H}}

 \def\grl{\mathfrak{l}}

\def\grP{\mathfrak{P}}

\begin{document}
\baselineskip 16pt

\medskip
\begin{center}
\begin{Large}\fontfamily{cmss}
\fontsize{17pt}{27pt}
\selectfont
	\textbf{Isotropic Grassmannians, Pl\"ucker and Cartan maps}
	\end{Large}
	
\bigskip \bigskip
\begin{large}  F. Balogh$^{1, 2}$\footnote[1]{e-mail:balogh@crm.umontreal.ca}, 
J. Harnad$^{1, 3}$\footnote[2]{e-mail:harnad@crm.umontreal.ca}  
and J. Hurtubise$^{1, 4}$\footnote[3]{e-mail:jacques.hurtubise@mcgill.ca}
 \end{large}
 \\
\bigskip
\begin{small}
$^{1}${\em Centre de recherches math\'ematiques, Universit\'e de Montr\'eal, \\C.~P.~6128, succ. centre ville, Montr\'eal, QC H3C 3J7  Canada}\\
$^{2}${\em John Abbott College, Ste.~Anne de Bellevue, QC H9X 3L9 Canada }\\
$^{3}${\em Department of Mathematics and Statistics, Concordia University\\ 1455 de Maisonneuve Blvd.~W.~Montreal, QC H3G 1M8  Canada}\\
$^{4}${\em Department of Mathematics and Statistics, McGill University, \\ 805 Sherbrooke St.~W.~Montreal, QC  H3A 0B9 Canada }
\end{small}
 \end{center}
\medskip
\begin{abstract}
This work is motivated by the relation between the KP and BKP integrable hierarchies,
whose $\tau$-functions may be viewed as flows of sections of dual determinantal and Pfaffian line bundles
over infinite dimensional Grassmannians.  In finite dimensions,  we show how to relate the {\em Cartan map} 
which, for a vector space $V$ of dimension $N$,  embeds the  Grassmannian  $\Gr^0_V(V+V^*)$ 
of maximal isotropic subspaces of $V+ V^*$,  with respect to the natural scalar product,  
into the projectivization of the exterior space $\Lambda(V)$,  and the {\em Pl\"ucker map}, 
which embeds the Grassmannian $\Gr_V(V+ V^*)$  of all $N$-planes in $V+ V^*$  into the projectivization 
of $\Lambda^N(V + V^*)$. The  Pl\"ucker coordinates on $\Gr^0_V(V+V^*)$ are expressed 
bilinearly in terms of the {\em Cartan coordinates}, which are  holomorphic sections of the dual
 Pfaffian line bundle  $\Pf^* \ra \Gr^0_V(V+V^*, Q)$. In terms of affine coordinates on the big cell, 
 this is equivalent to an identity of  Cauchy-Binet type, expressing the determinants of square submatrices 
 of a skew symmetric $N \times N$  matrix as bilinear sums over the Pfaffians of their principal minors. 
 \end{abstract}

\section{Introduction: $\tau$-functions for integrable hierarchies and Grassmannians} 
\label{KP_BKP_tau}

We recall the relation between integrable hierarchies and infinite dimensional Grassmannians
developed by Sato and the Kyoto school \cite{Sa, DJKM1, DJKM2}.  (For expository accounts, see \cite{JM, Dick, HB}.)
Solutions to the  KP (Kadomtsev-Petviashvili) hierarchy can be expressed in terms of KP $\tau$-functions 
$\tau_w({\bf t})$, parametrized by elements $w \in \Gr_{\HH_+}(\HH)$ of an infinite Grassmannian consisting
of subspaces $w \ss \HH$ of a polarized Hilbert space $\HH=\HH_+ + \HH_-$, commensurate with
the subspace $\HH_+\ss \HH$,  and depending on an infinite sequence of commuting flow variables
\be
{\bf t} = (t_1, t_2, \dots).
\ee
These satisfy the Hirota bilinear residue relations,
\be
\res_{z=0}\left(e^{\sum_{i=1}^\infty (t_i-s_i)z^i} \tau_w({\bf t} - [z^{-1}])\tau(\sb + [z^{-1}])  \right)dz  = 0,
\label{hirota_bilinear_tau_res}
\ee
identically in ${\bf s}$,  where
\be
{\bf s} :=(s_1, s_2, \dots),  \quad [z^{-1}] := \left({1\over z}, {1\over 2z^2},  \dots,{1\over j z^j}, \dots \right).
\ee

Expanding $\tau_w({\bf t})$  in a basis of Schur functions \cite{Mac, Sa}
\be
\tau_w({\bf t}) = \sum_{\lambda}\pi_\lambda(w)s_\lambda({\bf t}),
\label{tau_schur_exp}
\ee
with the flow parameters $(t_1, t_2, \dots)$ interpreted as normalized power sums
\be
t_i = {p_i\over i} \quad i=1,2 \dots,
\ee
and  the labels $\lambda$ denoting integer partitions $\lambda =(\lambda_1 \ge \lambda _2 \ge\cdots \ge \lambda_{\ell(\lambda)} >0, \cdots)$
of  length $\ell(\lambda)$ and weight $|\lambda|$,  the coefficients $\pi_\lambda(w)$ may be  interpreted 
as {\em Pl\"ucker coordinates} of the element  $w \in \Gr_{\HH_+}(\HH)$.
 
 The Hirota equation (\ref{hirota_bilinear_tau_res}) is then
formally equivalent to the infinite set of {\em Pl\"ucker relations}
\be
\sum_i(-1)^{i+\mu(\lambda_i-i+1)}\pi_{[\lambda^{-},\lambda_i]}(w)\pi_{[\mu^+,\lambda_i-i+1]}(w) =0,
\label{Plucker_relations_explicit}
\ee
where $(\lambda,\mu)$  is any pair of partitions,
$[\lambda^{-},\lambda_i]$ is the partition of length $\ell(\lambda)-1$ with parts
\be
[\lambda^{-},\lambda_i] := (\lambda_1+1,\lambda_2+1,\dots,\lambda_{i-1}+1,\lambda_{i+1},\dots,\lambda_{\ell(\lambda)}),
\label{Plucker_lambda_minus}
\ee
and $[\mu^{+},u]$ is the partition of length $\ell(\mu) +1$ with parts
\be
[\mu^{+},u]:=\begin{cases}
(u+1,\mu_1,\mu_2,\mu_3,\dots,\mu_{\ell(\mu)}) & \text{if } u> \mu_1-2\cr
(\mu_1-1,u+2,\mu_2,\mu_3,\dots,\mu_{\ell(\mu)}) & \text{if } \mu_2-3<u<\mu_1-2\cr
(\mu_1-1,\mu_2-1,u+3,\mu_3,\dots,\mu_{\ell(\mu)}) & \text{if } \mu_3-4<u<\mu_2-3\cr
\vdots & \vdots \cr
(\mu_1-1,\mu_2-1,\mu_3-1,\dots,\mu_{\ell(\mu)-1}-1,u+\ell(\mu)) & \text{if } u<\mu_{\ell(\mu) -1}-\ell(\mu),
\end{cases}
\label{Plucker_mu_plus}
\ee
which is not defined if $u=\mu_k-k-1$ for some $1\leq k \leq \ell(\mu)$. The summands corresponding to indices $i$ in \eqref{Plucker_relations_explicit} for which $[\mu^+,\lambda_i-i+1]$ is not defined are 
omitted from the sum.
 The exponent $\mu(u)$ in (\ref{Plucker_relations_explicit}) is defined as the position $k$ 
 of the inserted part $u+k$ in \eqref{Plucker_mu_plus}.

The Pl\"ucker relations \eqref{Plucker_relations_explicit} determine the image of the infinite Grassmannian
 $\Gr_{\HH_+}(\HH)$
under the {\em Pl\"ucker map}:
\bea
 \Pl_{\HH_+, \HH}: \Gr_{\HH_+}(\HH) &\&\ra \Pb(\FF) \cr
\Pl_{\HH_+, \HH}: \Span\{w_1, w_2, \dots \} &\&\mapsto \left[w_1 \wedge w_2 \wedge \cdots \right] 
=\big[\sum_\lambda \pi_\lambda(w) |\lambda; N\rangle\big] \in \Pb(\FF) ,
\eea
embedding $\Gr_{\HH_+}(\HH)$ into the projectivization of the semi-infinite
wedge \hbox{product} space  (the fermionic Fock space)
\be
\FF = \Lambda^{\infty/2}(\HH) =\sum_{N\in \Zb}\FF_N.
\ee
Here $\{|\lambda;N\rangle\}$ is the standard basis \cite{Sa, SW, HB} for the fermionic charge-$N$ sector
$\FF_N$ of the Fock space $\FF$ and $\{w_1, w_2, \dots \} $ is any admissible basis for the
subspace $w \ss \HH$, viewed as an element of $Gr_{\HH_+}(\HH)$. 
As in the finite dimensional case (detailed below in Sections  \ref{Det_Pfaff_line_bundles} and \ref{plucker_cartan}),
 the Pl\"ucker coordinates $\{\pi_\lambda(w)\}$  are expressible as determinants of suitably defined matrices,
 $W_\lambda(w)$, which are maximal minors of the homogeneous coordinate matrix $W(w)$ of the element $w$.
They may be interpreted  as  holomorphic sections of the (dual) determinantal line bundle  $\Det^*\ra \Gr_{\HH_+}(\HH)$ (see \cite{SW, HB}).

The BKP hierarchy \cite{DJKM1, DJKM2} may similarly be expressed in terms of a BKP $\tau$-function $\tau_{w^0}^B({\bf t}_B)$,
parametrized by elements $w^0 \in \Gr^0_{\HH_+}(\HH^B, Q)$ of the  Grassmannian of maximal isotropic subspaces 
of a Hilbert space $\HH^B$ endowed with a suitably defined complex scalar product $Q$. 
It depends only on the odd flow variables
\be
{\bf t}_B = (t_1, t_3, \dots )
\ee
and satisfies the Hirota bilinear residue equation
\be
\res_{z=0}\left(e^{\tilde{\xi}(z, {\bf t}_B - \sb_B)} \tau_{w^0}^B({\bf t}_B - 2[z^{-1}]_B)\tau_{w^0}^B(\sb_B + 2[z^{-1}]_B){dz\over z} \right)
= \tau_{w^0}^B({\bf t}_B)\tau_{w^0}^B(\sb_B),
\label{hirota_bilinear_tau_res_BKP}
\ee 
identically in
\be
{\bf s}_B = (s_1, s_3, \dots),
\ee
where
\be
\tilde{\xi}(z, {\bf t}_B) := \sum_{j=1}^\infty t_{2j -1} z^{2j -1}, \quad
[z^{-1}]_B := \left(z^{-1}, {1\over 3} z^{-3}, {1\over 5}  z^{-5}, \dots \right).
\ee

We may similarly expand $\tau^B_{w^0}({\bf t}_B)$ in a series \cite{DJKM1,  You, Shig1, Shig2}
\be
\tau_{w^0}^B({\bf t}_B) = \sum_{\alpha \, \in \{ \text{even strict partitions}\}} \kappa_\alpha (w^0)\QQ_\alpha({\bf t}_B)
\label{Q_schur_expansion}
\ee
where, up to normalization
\be
\QQ_{\alpha}({\bf t}_B):=\frac{1}{\sqrt{2^{r}}}Q_{\alpha}\left(\frac{{\bf t}_B}{2}\right)
\ee
are Schur's $Q$-functions  \cite{Mac} (also known as projective Schur functions),
 labelled by strict partitions $\alpha = (\alpha_1 > \alpha_2 > \cdots \> \alpha_r \ge 0)$.
 The coefficients $\{\kappa_\alpha(w^0)\}$ may be interpreted as Pfaffians of skew symmetric matrices 
 $A^\emptyset(w^0)(\alpha)$,
also labelled by strict partitions, which are, up to projectivization,  principal minors of the skew symmetric affine 
coordinate matrix $A^\emptyset(w^0)$ representing the image of  $w^0\in \Gr_{\HH}^0(\HH^B)$ within the ``big cell'',  under the
{\em Cartan map} (\cite{Ca, HS1, HS2} and \cite{HB}, Chapt. 7,  Appendix E)
\be
\Ca_\HH: \Gr^0_\HH(\HH^B) \ra \Pb(\FF^B)
\ee
embedding the isotropic Grassmannian $\Gr^0_\HH(\HH^B)$ into the projectivization of the ``neutral fermion'' Fock space $\FF^B$ \cite{DJKM1, DJKM2}. 
(See Section \ref{Det_Pfaff_line_bundles}, eq.~(\ref{Cartan_map_def}) for the definition in finite dimensions.)

The coefficients  $\{\kappa_\alpha(w^0)\}$,  which we refer to as {\em Cartan coordinates} (cf. \cite{HB}, Chapt. 7), are similarly 
interpreted as sections of a holomorphic line bundle: the dual Pfaffian line bundle $\Pf^* \ra \Gr^0_\HH(\HH^B)$ 
over the isotropic Grassmannian $\Gr^0_\HH(\HH^B)$, which will be defined below in the finite dimensional setting
originally studied by Cartan. They also satisfy quadratic relations that determine the image of  $\Gr^0_\HH(\HH^B)$
under the Cartan map,  the {\em Cartan relations} (\cite{Ca,  DJKM1, HS1, HS2} and \cite{HB}, Chap.~7 and Appendix C)
(or the {\em Pfaffian Pl\"ucker relations}, as they are called in \cite{Oh, Shig2}):
\bea
&\&\sum_{i=1}^{\ell(\alpha)} (-1)^{i + \beta(\alpha_i)}  \kappa_{(\alpha^-, \alpha_i)}\kappa_{(\beta^+, \alpha_i)} + 
\sum_{i=1}^{\ell(\beta)} (-1)^{i +\alpha(\beta_i)} \kappa_{(\beta^-, \beta_i)}\kappa_{(\alpha^+, \beta_i)}\cr
&\& =\frac{1}{2}\left((-1)^{\ell(\alpha)+\ell(\beta)} -1\right)\kappa_{\alpha}\kappa_{\beta}.
\label{BKP_cartan_equations}
\eea
Here $(\alpha, \beta)$ is any pair of strict partitions of lengths $(\ell(\alpha),\ell(\beta))$.
For a strict partition $\alpha = (\alpha_1 > \alpha_2 > \cdots > \alpha_{\ell(\alpha)}\ge 0)$
and any nonnegative integer $m$ lying between $\alpha_i$ and $\alpha_{i+1}$:
\be
\alpha_i > m > \alpha_{i+1},
\label{ordering_alpha_i}
\ee
 $(\alpha^+, m)$ is defined to be the strict partition of length $\ell(\alpha)+1$ obtained from $\alpha$ by adding the part $m$:
\be
(\alpha^+, m) := (\alpha_1, \dots, \alpha_i, m, \alpha_{i+1}, \dots \alpha_{\ell(\alpha)})
\ee 
while, for any $\alpha_i$, $(\alpha^-, \alpha_i)$ is defined as the strict partition\index{partitions} of length
$\ell(\alpha)-1$ obtained from $\alpha$ by omitting the part $\alpha_i$:
\be
(\alpha^-, \alpha_i) = (\alpha_1, \dots, \alpha_{i-1}, \alpha_{i+1}, \dots , \alpha_{\ell(\alpha)}).
\ee
For $m\in \Nb$, the integers $\alpha(m), \beta(m)$ are defined to be  the number of parts of $\alpha$ and $\beta$,
respectively, greater than $m$.

It is a standard result \cite{DJKM1, DJKM2} that the square $(\tau_{w^0}^B({\bf t}_B))^2$ of a BKP $\tau$-function is equal 
to the restriction of a KP $\tau$-function $\tau_w({\bf t})$ to the values
\be
{\bf t}' := (t_1, 0, t_3, 0, t_5, 0, \dots ).
\ee
This implies an identity expressing the Pl\"ucker coordinates $\pi_\lambda(w^0)$ quadratically in terms of the Cartan coordinates $\kappa_\alpha(w^0)$.
 In the setting of finite dimensional Grassmannians, this quadratic relation is given  by Theorem \ref{main_theorem}  
 of Section \ref{Det_Pfaff_line_bundles}. It is closely related to Cartan's identification  of maximal isotropic subspaces 
 of a complex Euclidean vector space with projectivized {\em pure spinors} \cite{Ca}.

For a vector space $V$ of dimension $N$, the dual determinantal and Pfaffian line bundles 
\be
\Det^*\ra \Gr_V(V+V^*)\quad \text{ and }\quad  \Pf^*\ra \Gr^0_V(V + V^*)
\ee
over the  Grassmannian $\Gr_V(V+V^*)$ of $N$-dimensional 
subspaces of the $2N$-dimensional space $V + V^*$ and the maximal isotropic Grassmannian $\Gr^0_V(V + V^*)$
with respect to the tautologically defined scalar product, respectively,  are defined in Section \ref{Det_Pfaff_line_bundles},
as well as the Pl\"ucker and Cartan maps embedding these into $\Pb(\Lambda^N(V+V^*)$ and $\Pb(\Lambda(V))$, respectively,
\bea
\Pl_V: \Gr_V(V+ V^*)&\& \hookrightarrow  \Pb(\Lambda^N(V+V^*), 
\label{plucker_map1}
 \\
&\& \cr
\Ca_V: \Gr^0_V(V+ V^*)&\& \hookrightarrow  \Pb(\Lambda(V)),
\label{cartan_map1}
\eea
with images determined  by the {\em Pl\"ucker relations} and 
the {\em Cartan relations}, respectively.

Section \ref{Det_Pfaff_line_bundles} reviews the construction of  (dual) determinantal and Pfaffian lines bundles 
$\Det^* \ra \Gr_V(V + V^*)$ and $\Pf^*_{\pm} \ra \Gr^{0\pm}_V(V + V^*)$ over the Grassmannian
$\Gr_V(V + V^*)$ of $N$-planes in the direct sum of a (complex) $N$-dimensional vector space $V$ and its dual $V^*$, and
the Grassmannian $\Gr^{0\pm}_V(V + V^*)$ of maximal isotropic subspaces of $V + V^*$ with
respect to the canonically defined scalar product $Q$ associated to the dual pairing, respectively. The latter  is related to
the hyperplane section bundle over the (projectivized) irreducible Clifford module associated to
the  scalar product $Q$, and its connected components over the irreducible $1/2$-spinor modules. Theorem \ref{main_theorem} gives
the main result, expressing the Pl\"ucker map (\ref{plucker_map1}) bilinearly in terms of the Cartan map (\ref{cartan_map1}),
 and thereby effectively inverting the Cartan embedding. Section \ref{proof_main_theorem} provides a direct proof of this theorem, 
 without the use of Pl\"ucker or Cartan coordinates.
 Proposition \ref{plucker_iota_cartan}, Section \ref{plucker_k_N_factoriz}, gives a factorization of the Pl\"ucker map 
 \be
 \Pl_k(V) :\Gr_k(V)\ \ra \Pb(\Lambda^k(V))
 \ee
  as the composition of  the tautological embedding map: 
 \be
 \iota_V:\Gr_k(V)\ra \Gr^0_V(V+ V^*)\
 \ee
 with the Cartan map.  The Cartan  coordinates $\kappa_\alpha(w^0)$  are defined  in Section \ref{plucker_cartan}  and 
 expressed as  Pfaffians of principal  minors of the affine coordinate matrix  $A^\emptyset(w^0)$ on 
the ``big cell'' of the isotropic  Grassmannian $\Gr^0_V(V+ V^*)$.  Theorem \ref{pfaffian_CB_identity}, Section \ref{pfaffian_CB},  interprets 
 Theorem \ref{main_theorem} in coordinate form as a Pfaffian analog of the Cauchy-Binet identity \cite{Mac} and gives
  an alternative proof, using inner and outer products on the exterior algebra $\Lambda(V+V^*)$.
  \br
  A number of other identities relating Pfaffians and determinants formed from skew matrices have been studied in the literature (see e.g. \cite{Ok} 
and references therein), but  none of these seem to coincide with the results of Theorems \ref{main_theorem} and \ref{pfaffian_CB_identity}.
  \er

\section{Pl\"ucker and Cartan maps:  determinantal and \\ Pfaffian line bundles}
\label{Det_Pfaff_line_bundles}

Let $V$ be a complex vector space of dimension $N$, $V^*$ its dual space and
$\Gr_V(V+V^*)$  the Grassmannian of $N$-planes in $V+ V^*$.
The {\em Pl\"ucker map}  \cite{GH}
\be
\Pl_V:\Gr_V(V +V^*) \ra \Pb(\Lambda^N(V+V^*)) 
\ee
is  the $\GL(V+ V^*)$ equivariant embedding  of $\Gr_V(V+V^*)$ 
in the projectivization $ \Pb(\Lambda^N(V+V^*))$ of the exterior space $\Lambda^N(V+V^*)$ 
defined   by:
\be
\Pl_{V} : w \mapsto [w_1 \wedge \cdots \wedge w_N] \in \Pb(\Lambda^N(V+V^*)),
\ee
where  $\{w_1, \dots, w_N\}$ is a basis for the subspace $w\in \Gr(V+V^*)$.
 Its image is cut out by the intersection
of a number of quadrics, the {\em Pl\"ucker quadrics}, defined by the {\em Pl\"ucker relations} (\cite{GH} Chapt. I.5 
and eq.~(\ref{Plucker_relations_explicit}) above).

 Let $\{e_i\}_{i=1, \dots, N}$  be a chosen basis for $V$   and $\{f_i\}_{i=1, \dots, N}$ the dual basis for $V^*$,
   \be
  f_i(e_j)= \delta_{ij},
  \ee
  Denote by $(e_1, \dots, e_{2N})$ the basis for $V+ V^*$ in which
\be
e_{N+i} := f_i, \quad 1\le i \le N.
\ee 
 Define the basis $\{|\lambda\rangle\}$ for $\Lambda^N(V+V^*)$  by
\be
|\lambda\rangle := e_{l_1} \wedge \cdots \wedge e_{l_N},
\ee
where $\lambda$ is any partition whose Young diagram fits in the $N\times N$ square diagram,
 and
\be
l_j := \lambda_j -j +N +1, \quad 1\le j \le N
\ee
 are the {\em particle positions}  (see, e.g., \cite{HB}, Chapt. 11, Sec. 11.3) associated to the partition
 \be
 \lambda = (\lambda_1, \dots, \lambda_{\ell(\lambda)}, 0, \dots ).
 \ee
 Thus $l_1 > \cdots > l_N$ is a strictly decreasing sequence of positive integers between $1$ and $2N$.
The (complex) scalar product on $\Lambda^N (V + V^*)$ is defined, in bra/ket notation,
 by requiring the  $\{|\lambda\rangle\}$ basis to be orthonormal
\be
\langle\lambda | \mu\rangle  =\delta_{\lambda\mu}.
\ee

 Following Cartan \cite{Ca}, we define the natural complex scalar product $Q$ on $V + V^*$ by
\be
Q((X, \mu), (Y, \nu)) = \nu(X)  + \mu(Y),
\ee
where $X, Y \in V$, $\mu, \nu \in V^*$.
The standard irreducible representation 
\bea
\Gamma: \CC_{V+V^*, Q}&\&\ra  \End(\Lambda(V)),  \cr
\Gamma: \sigma &\&\mapsto \Gamma_\sigma
\eea
of the Clifford algebra $\CC_{V+V^*,Q}$ on $V+ V^*$ determined by 
the scalar product $Q$ is generated by the linear elements, which are defined by exterior and interior multiplication:
\bea
\Gamma_{v + \mu}:= v\wedge &\&+ i_{\mu}  \  \in \End(\Lambda(V))  \cr
v \in V, &\& \quad \mu\in V^*.
\label{Gamma_w}
\eea

The Clifford representations of the basis elements are denoted 
\be
\psi_i := \Gamma_{e_i} = e_i \wedge, \quad \psi^\dag_i:= \Gamma_{f_i} = i_{f_i}, \quad i=1, \dots, N,
\ee
and viewed as finite dimensional fermionic creation and annihilation operators,
satisfying the anticommutation relations 
\be
[\psi_i, \psi_j]_+=0,\quad [\psi^\dag_i, \psi^\dag_j]_+=0, \quad [\psi_i, \psi^\dag_j]_+= \delta_{ij}.
\ee

Let
\be
\Gr^0_V(V+V^*) \ss \Gr_V(V+V^*)
\ee
be the sub-Grassmannian of $N$-planes in $V + V^*$ that are totally isotropic with respect to $Q$.
That is, if $\{w_1, \dots, w_N\}$ is a basis for an element $w^0\in \Gr^0_V(V+V^*)$, then
\be
Q(w_i, w_j) =0, \quad 1\le i,j \le N.
\label{null_quad_rels}
\ee
It  then follows from (\ref{null_quad_rels}) and the Clifford algebra relations that
\be
\Gamma_{w_i} \Gamma_{w_j} + \Gamma_{w_j} \Gamma_{w_i}  =0, \quad 1\le i,j \le N.
\ee
Together with the linear independence of the $w_j$'s, this implies  \cite{Ca} that
\be
\rank\left(\prod_{j=1}^N \Gamma_{w_j}\right) =1.
\ee

\begin{definition}
\label{cartan_map}
The {\em Cartan map}  $\Ca_V: \Gr^0_V(V+V^*)\ra \Pb(\Lambda(V))$,
defined  \cite{Ca, HS1, HS2} by:  
\be
\Ca_V : w^0\mapsto \Image\left(\prod_{j=1}^N \Gamma_{w_j}\right) \in \Pb(\Lambda V).
\label{Cartan_map_def}
\ee
gives an equivariant embedding of the isotropic  Grassmannian $\Gr_V^0(V+V^*)$  
into the projectivization $\Pb(\Lambda(V))$ of the exterior space $\Lambda(V)$ (the irreducible Clifford module).
\end{definition}
It intertwines the action of the orthogonal group $SO(V+V^*,Q )$  on $\Gr^0_V(V+V^*)$
with the (projectivized) representation of the spin group $\Spin(V+V^*, Q)$
on $\Pb(\Lambda(V))$ determined by the Clifford representation.  Its image $\Ca_V(\Gr^0_V(V + V^*))$  
consists of  (the projectivization of) {\em pure spinors} \cite{Ca, Ch}, which are the elements of the lowest dimensional
stratum of the $\Spin(V+V^*, Q)$ representation on $\Lambda(V)$.
Similarly to the Pl\"ucker map, its image is cut out by a set of homogeneous quadratic relations, 
which we refer to as the {\em Cartan relations}.  
(Cf. \cite{Ca}, Secs. 106-107,  eq.~(\ref{BKP_cartan_equations}) above, and \cite{HB}, Appendix E.)

The irreducible Clifford module  $\Lambda (V)$ is the direct sum of  the two
 irreducible ${1\over 2}$-spinor modules (Weyl spinors)
 \be
 \Lambda(V) = \Lambda_+(V) \oplus \Lambda_-(V),
 \label{spinor_mod_sum}
 \ee
 consisting of even $(+)$ and odd $(-)$ multivectors $\vb \in \Lambda(V)$,  denoted  $\vb_+$ and  $\vb_-$, respectively.
The isotropic Grassmannian $\Gr^0_V(V+V^*)$ is the disjoint union of
 two connected components
 \be
 \Gr^0_V(V+V^*) = \Gr^0_{V_+}(V+V^*)\sqcup  \Gr^0_{V_-}(V+V^*),
 \ee 
such that
\be
\Ca_V(w^0)(\Gr^0_{V_\pm}(V+V^*)) \ss \Pb(\Lambda_\pm(V) ).
\ee
In particular the Cartan image of the element $V\in \Gr_V^0(V+V^*)$ is
\be
\Ca_V(V) = [\Omega_V],
\label{cartan_V}
\ee
where $\Omega_V \in \Lambda^NV$ is a volume form on $V$.

We then have six standard holomorphic line bundles:
\begin{enumerate}
\item The hyperplane section bundle $\OO(1) \ra \Pb(\Lambda^N(V+V^*))$  dual to the tautological bundle.
\item The pair of hyperplane section bundles $\OO(1) \ra \Pb(\Lambda_\pm(V))$.
\item The dual determinantal line bundle $\Det^* \ra \Gr_V(V + V^*)$.
\item The pair of dual Pfaffian line bundles $\Pf^*_{\pm} \ra \Gr^{0}_{V_\pm}(V + V^*)$.
\end{enumerate}

Bundles  $1$ and $3$,  and $2$ and $4$ are related by pullback under the Pl\"ucker and Cartan maps, respectively:
\bea
\Det^*&\&= \Pl_V^*\left(\OO(1)(\Pb(\Lambda^N(V+V^*)))\vert_{\Pl_V(\Gr_V(V+V^*)})\right), \cr
\Pf^*_{\pm} &\&= \Ca_V^*\left(\OO(1)(\Pb(\Lambda_{\pm}(V)))\vert_{\Ca_V(\Gr^{0\pm}_V(V+ V^*))}\right).
\eea
The dimensions of their spaces of holomorphic sections are
\bea
h^0(\Gr_V(V+V^*), \Det^*)&\&=h^0(\Pb(\Lambda^N(V + V^*)), \OO(1)(\Pb(\Lambda^N(V + V^*)))=\left( 2N \atop N\right), \cr
h^0(\Gr^{0}_{V_\pm}(V+V^*)), \Pf^*_{\pm}) &\&=h^0(\Pb(\Lambda_{\pm}(V)), \OO(1)(\Pb(\Lambda_{\pm}(V))))= 2^{N-1}.
\eea
Bundles $3$ and $4$ are related by
\be
(\Pf^*_\pm)^{\otimes 2} \simeq\Det^*\big |_{\Gr^0_{V\pm}(V+ V^*)}.
\label{Pfaff_2_Det}
\ee

From the viewpoint of geometric representation theory (\cite{Fu}, Chapt.~9),
the space $\Lambda^N(V + V^*)$ is the representation module  of  $\SL(V+V^*, \Cb)$ (and, by restriction, $\Spin(V+ V^*, Q)$)
 obtained from the Borel-Weil theorem by identifying it with the space of holomorphic sections 
 of the line bundle $\Det^* \ra \Gr_V(V+V^*)$. Denote by $V_{\pm}$ the subspaces: 
 \bea
 V_+    &\&:=\begin{cases} V=\span(e_1, \dots , e_N)\in \Gr^0_{V_+}   \text{ if }  N \text{ is even}, \\
 \span(e_1, \dots, e_{N-1}, f_N) \in \Gr^0_{V_+}    \text{ if }  N \text{ is odd} , \end{cases} \cr
 V_-   &\&:=\begin{cases} V= \span(e_1, \dots , e_N)  \in \Gr^0_{V_-} \text{ if } N \text{ is odd}, \\\
  \span(e_1, \dots, e_{N-1}, f_N) \in \Gr^0_{V_-}   \text{ if } N \text{ is even}. \end{cases}
 \eea
 Let $P^{\pm} \ss\SO(V+ V^*, Q)$ be the stabilizers of $V_{\pm}$, and
 $\tilde{P}^{\pm} \ss \Spin(V+ V^*, Q)$ the corresponding stabilizers of their
 images $\Ca(V_{\pm}) \ss \Lambda_{\pm} (V)$ under the Cartan map.
 The components $\Gr^0_{V_\pm}(V+ V^*) $ are the orbits of $V_\pm$ under  $ \SO(V+ V^*, Q)$,
 and hence are identified with the homogeneous spaces 
\bea
 \Gr^0_{V_+}(V+ V^*) &\&= \SO(V+ V^*, Q)/P^+ = \Spin(V+ V^*, Q)/\tilde P^+\\
  \Gr^0_{V_-}(V+ V^*)&\&= \SO(V+ V^*, Q)/P^- = \Spin(V+ V^*, Q)/\tilde P^-.
\eea

As a representation of $\SO(V+ V^*, Q)$ (and  $\Spin(V+ V^*, Q)$), the  module
$\Lambda^N(V + V^*)$ decomposes into the direct sum of two irreducible modules
\be
\Lambda^N(V + V^*) = \Lambda^N_+(V + V^*) \oplus \Lambda^N_-(V + V^*), 
\ee
such that
\be
\Pl(V_{\pm}) \ss \Pb(\Lambda^N_\pm(V + V^*)).
\ee
Let $\Det^*_\pm$ be the restriction of $\Det^*$ to $\Gr^0_{V_\pm}(V+ V^*)\subset \Gr_V(V+ V^*)$. Then
\be
\Pl(\Gr^0_{V_ \pm}(V+V^*))\ss \Pb(\Lambda^N_\pm (V + V^*))
\ee
 and
 \be
\Det^*_\pm= \Pl_V^*\left(\OO(1)(\Pb(\Lambda^N_\pm(V+V^*)))\vert_{\Pl_V(\Gr_V(V+V^*)})\right).
\ee

The standard diagonal Cartan subalgebra $\grH$ of $\SO(V+V^*,Q)$ is isomorphic to $V$, by the map $v\mapsto diag(v, -v)$.  We can take  the element
\be
e_1\wedge e_2\wedge \cdots \wedge e_N \in \Lambda^N_{(-1)^{N}}(V) \subset  \Lambda^N_{(-1)^{N}}(V + V^*)
\ee
 as a generator of the highest weight space for $ \Lambda^N_{(-1)^{N}}(V)$, and  
 \be
 e_1\wedge e_2\wedge \cdots \wedge e_{N-1}\wedge f_N  \in  \Lambda^N_{(-1)^{N+1}}(V + V^*)
 \ee
 as a generator of the highest weight space for 
$ \Lambda^N_{(-1)^{N+1}}(V)$.
 The highest weight on $\grH$ corresponding to $e_1\wedge e_2\wedge..\wedge e_N$ is then  $(f_1+...+f_N)$, and the highest weight on $\grH$ corresponding to $e_1\wedge e_2\wedge..\wedge e_{N-1}\wedge f_N$ is  $(f_1+...+f_{N-1}- f_N)$. 
  If $\chi_{\pm}$ are the characters on the Cartan subgroup associated to these weights, the line bundles  $\Det^*_\pm$ over the homogeneous spaces $\Gr^0_{V_\pm}(V+ V^*)=  \Spin(V+ V^*, Q)/\tilde P^\pm$  can be constructed in the standard way, as quotients by $\tilde{P}^\pm$  of the trivial bundle $ \Spin(V+ V^*, Q)\times_{\tilde P^\pm}\Cb \ra \Spin(V+ V^*, Q) $, 
 where the parabolic subgroups $\tilde{P}^\pm$ act  on $\Spin(V+ V^*, Q)$ by right multiplication and on $\Cb$ by
 multiplication by the inverses $\chi^{-1}_\pm$ of the characters, 

 Likewise, the spin module $\Lambda(V)$ decomposes into the 
sum  (\ref{spinor_mod_sum}) of the even and odd irreducible ${1\over 2}$-spinor modules  (Weyl spinors)
corresponding to sections of the line bundles $\Pf^*_\pm$ over the two 
connected components  $\Gr^0_{V_\pm}(V+V^*)$ of the isotropic Grassmannian. 
The highest weight of $\Lambda_{(-1)^N}(V)$ is $(f_1+ \dots +f_N)/2$, and that of  
$\Lambda_{(-1)^{N+1}}(V)$ is $(f_1+ \dots +f_{N-1}-f_N)/2$. If $\chi_{0,\pm}$ are the corresponding characters, the line bundles
 $\Pf^*_\pm$ are built in the same way as $\Det^*$, but using $\chi_{0,\pm}$  instead of $\chi_\pm$, and the identification 
 of $\Gr^0_{V_\pm}(V+ V^*)$ as  $\Spin(V+ V^*, Q)/\tilde P^\pm$.

The fact that the restriction of  $(\chi_{0,\pm})^2$ to $\Gr^0_{V_\pm}(V+V^*)$ is equal to $ \chi_\pm$,
\be
(\chi_{0,\pm})^2 \vert_{\Gr^0_V(V+V^*)^\pm}= \chi_\pm,
\ee
 then gives the isomorphisms (\ref{Pfaff_2_Det}).
 On the spaces of holomorphic sections of these bundles, i.e., the representation modules, there are corresponding maps
 \be
  \Lambda_\pm(V)\otimes  \Lambda_\pm(V)\rightarrow \Lambda^N_\pm(V + V^*),
  \ee
  which are defined by restricting the bilinear form $\beta_N$ defined below in eq.~\eqref{beta_N_def} 
  to the subspaces $\Lambda_{\pm}(V)$.

The main result of this paper is to show how the Pl\"ucker and Cartan maps are  related, and 
to express this relation in terms of  Pl\"ucker and  Cartan coordinates.
For an $r$-element subset $I= (I_1, \dots , I_r)\subseteq \{1, \dots , N\}$,  $0 \le r \le N$, arranged in decreasing order
$I_1 > \cdots > I_r$, we denote the complementary subset by  $\tilde{I} = (\tilde{I}_1, \dots, \tilde{I}_{N-r}) \subseteq \{1, \dots N\}$,
 also arranged in decreasing order $\tilde{I}_1 > \cdots > \tilde{I}_{N-r})$.
 
 \begin{definition} Let $\sgn(I)$  be the sign of the permutation $(I, \tilde{I})$ of $(1, \dots, N)$
 \be
 \sgn(I):= \sgn(I_1, \dotsm I_r, \tilde{I}_1, \dots, \tilde{I}_{N-r}).
  \ee
 \end{definition}
Recall the Hodge star automorphism $*: \Lambda(V) \ra \Lambda(V)$, defined relative to the basis $(e_1, \dots, e_{N})$,
as follows:
 \begin{definition} 
   $*:\Lambda(V) \ra \Lambda(V)$ is defined on basis elements by:
 \be
 *e_{I_1} \wedge \cdots \wedge e_{I_r} :=\sgn(I) \ e_{\tilde{I}_{1}} \wedge \cdots \wedge e_{\tilde{I}_{N-r}},
 \ee
  and  extended to $\Lambda(V)$ by linearity. 
  \end{definition}
  This coincides with the usual definition of $*$ with respect to
  the metric in which $(e_1, \dots, e_{N})$ is a positively oriented, orthonormal basis, and the volume form is
  \be
 \Omega_V := e_1 \wedge \cdots \wedge e_N\in \Lambda^N(V).
 \ee
 
We also introduce (following Cartan, \cite{Ca}),  the closely related  automorphism 
of $\Lambda(V)$ defined by forming the product of $N$ orthogonal linear elements
  \bea
  C:\Lambda(V) &\&\ra \Lambda(V) \cr
  C:\vb&\&\mapsto (\psi_1 - \psi^\dag_1)\cdots (\psi_N - \psi^\dag_N)\vb.
  \eea
 When acting on homogeneous elements  $\vb\in \Lambda^r(V)$, this is related to $*$ by
  \be
  C\vb= (-1)^{\tfrac{1}{2}r(r-1) +rN } *\vb.
  \label{C_Hodge_sign}
  \ee
   Up to a sign, its square is the identity element
  \be
  C^2 = (-1)^{\tfrac{1}{2}N(N+1)} \Ib,
  \label{C_squared}
  \ee
and it leaves invariant the scalar product  $(\vb, \wb)$ on $\Lambda(V)$ in which the basis elements
  \be
  e_I = e_{I_1} \wedge \cdots \wedge e_{I_r} , \quad N \ge I_1 > \cdots > I_r \ge 1,
  \ee
  are orthonormal:
  \bea
  (e_I, e_J) &\& :=  \delta_{IJ} \cr
  &\& \cr
   (C(\vb), C(\wb)) &\&= (\vb, \wb).
  \eea
 
 \begin{definition} 
 \label{def_beta_0_N}
Following Cartan  \cite{Ca},  we define the  scalar-valued bilinear form $\beta_0$ on $\Lambda(V)$ as
  \bea
  \beta_0: \Lambda(V) \times \Lambda(V) &\&\ra \Cb \cr
\beta_0(\vb,\wb) &\&=(\vb, C\wb),
\eea
and the  bilinear forms
 \be
 \beta_k: \Lambda(V) \times \Lambda(V) \ra \Lambda^N(V+V^*) , \quad k =1, \dots 2N
 \label{beta_N_def}
 \ee
with values in $\Lambda^k(V+V^*)$ such that, for $\sigma \in  \Lambda^{k}(V+V^*)$,
 \be
 \langle \sigma | \beta_k(\vb, \wb)\rangle := \beta_0(\vb, \Gamma_\sigma \wb) = (\vb, C\, \Gamma_\sigma \wb), \quad \vb, \wb \in \Lambda(V).
 \ee
  \end{definition} 
  Then, as shown by Cartan  \cite{Ca}, the image of the Cartan map \ref{Cartan_map_def} in $\Pb(\Lambda(V))$ is cut out 
  by the intersection of the quadrics 
 \be
  \beta_k(\Ca(w^0), \Ca(w^0)) = 0, \ k \equiv N \ \mod \ 4,  \ 0 \le k\le N-1,
  \label{cartan_bilinear_rels}
\ee
which are the {\em Cartan relations}.
\br
Although all the bilinear forms $\beta_k$ with $k \neq N$ vanish on the image of the
Cartan map, the ones with $k>N$ just duplicate those with $k <N$, and every second one is skew,
so the relations (\ref{cartan_bilinear_rels}) for other values of $k\neq N$
are either trivially satisfied or duplicated for the other cases.
\er
  \begin{definition}
  \label{diag_cartan}
The diagonal value of the product Cartan map is denoted
 \bea
 \Ca_V^{D}:\Gr^0_V(V+V^*) &\& \ra \Pb(\Lambda^N(V)) \times\Pb(\Lambda^N(V))   \cr
 \Ca_V^{D} (w^0) &\&  =  (\Ca_V(w^0), \Ca_V(w^0)).
 \eea
   \end{definition} 
 We then have our main result:
 \begin{theorem}
 \label{main_theorem}
 Up to projectivization, 
 \be
 \beta_N \circ \Ca_V^{D} = \Pl_V |_{\Gr^0_V(V+V^*)}.
 \label{beta_N_Ca_D_plucker}
 \ee
  \end{theorem}
 This may be deduced as a  consequence of results given by Cartan (\cite{Ca}, Sections 105-109)
  in somewhat implicit form.  A self-contained proof is given in  Section \ref{proof_main_theorem}.
 Theorem \ref{pfaffian_CB_identity}, Section \ref{pfaffian_CB}, gives an equivalent relation,
in coordinate form,  consisting of the bilinear  identity (\ref{pfaff_CB_ident}), which expresses the  determinants of all minors 
of any skew matrix in terms of the Pfaffians of its principal minors. The proof of this equivalence is given there,
as well as a direct proof of the identity (\ref{pfaff_CB_ident}).

\section{Proof of Theorem \ref{main_theorem} }
\label{proof_main_theorem}

\begin{proof}
Since  $C$ preserves the scalar product $(\ , \ )$ on $\Lambda(V)$, eq.~(\ref{beta_N_Ca_D_plucker}) is equivalent to
\be
(C(\Ca_V(w^0)_, C\, \Gamma_\sigma C(\Ca_V(w^0)))) = \langle \sigma | \Pl_V(w^0)\rangle, \quad \forall \sigma\in\Lambda^N(V+V^*).
\label{beta_sigma_Ca_Ca_dual}
\ee
However, as noted in \cite{Ca}, the transpose of $\Gamma_\sigma$ with respect to $( \ , \ )$ is
\be
\Gamma^T_\sigma = C \,\Gamma_\sigma \,C 
\ee
so this is equivalent to
\be
(\Gamma_\sigma C( \Ca_V(w^0)),  \Ca_V(w^0))) =  \langle \sigma | \Pl_V(w^0)\rangle .
\label{beta_sigma_Ca_Ca_dual}
\ee
It is sufficient to prove this for all basis elements of the form
\be
\sigma =f_I \wedge *e_{J},
\label{sigma_f_I_star_e_J}
\ee
for  pairs  $(I, J)$ of decreasingly ordered subsets $I, J \subseteq \{1, \dots, N\}$ of equal cardinality 
\be
r =|I|=|J|, \quad r=0, \dots, N.
\ee
For these,
(\ref{beta_sigma_Ca_Ca_dual}) becomes
\be
(\AA(\Gamma_{f_I}\Gamma_{*e_{J}}) C(\Ca_V(w^0)), \Ca_V(w^0))) = 
\langle f_I \wedge *e_{J} | \Pl_V(w^0)\rangle ,
\label{gamma_f_I_gamma_e_star_J_Ca_diag}
\ee
 where, for $\{w_i \in V+V^*\}_{i=1, \dots m}$
 \be
\AA(\Gamma_{w_1} \cdots \Gamma_{w_m}) := {1\over m!} \sum_{\tau \in \SS_m}\sgn(\tau) \Gamma_{w_{\tau(1)}} \cdots  \Gamma_{w_{\tau(m)}}
\label{antisym_def}
\ee
is the antisymmetrization map. Since $\Gamma_{f_I}$ and $\Gamma_{*e_{J}}$ are already antisymmetric, we have
\be
\AA(\Gamma_{f_I}\Gamma_{*e_{J}}) = \Gamma_{f_I \wedge *e_{J}}= \Gamma_\sigma.
\label{Gamma_FI_GJ_antisym}
\ee 
Therefore (\ref{gamma_f_I_gamma_e_star_J_Ca_diag}) is equivalent to
\be
(\Gamma_{f_I\wedge *e_{J}} C(\Ca_V(w^0)), \Ca_V(w^0))) = 
\langle f_I \wedge *e_{J} | \Pl_V(w^0)\rangle .
\label{gamma_f_I_wedge_e_star_J_Ca_diag}
\ee

By the equivariance of the Cartan and Pl\"ucker maps under
  $\Spin(V+V^*, Q)$ and $SO(V+V^*, Q)$, it is sufficient to prove this for just
one element $w^0$ in each of the two connected components of $\Gr^0_V(V+V^*)$. 
These may be chosen as  
\be
w^0= V=\span\{e_1, \dots, e_N\},
\ee
which lies in the component $\Gr^{0 }_{V_{(-1)^N}}(V+V^*)$, and 
 \be
 \tilde{w}^0= \span\{e_1, \dots, e_{N-1}, f_{N}\},
 \ee
which lies in $\Gr^{0 }_{V_{(-1)^{N+1}}}(V+V^*)$. 

For the first, we have
\be
\Ca_V(w^0)=\Ca_V(V) = [\Omega_V], \quad C(\Ca_V(w^0))= [1], \quad \Pl_V(w^0)=  [\Omega_V],
\ee
and
\be
( \Gamma_{f_I \wedge *e_{J}}1,  \Omega_V) 
 = \delta_{I, \emptyset} \delta_{J, \emptyset} = \langle f_I\wedge *e_{J} | \Omega_V \rangle.
\ee
For the second,
\be
\Ca_V(\tilde{w}^0 )  = [e_1\wedge \cdots \wedge e_{N-1}],  \quad  C(\Ca(\tilde{w}^0)) = [e_N], \quad
 \Pl_V(\tilde{w}^0)  = [e_1\wedge \cdots \wedge e_{N-1} \wedge f_{N}],
\ee
and
\be
(  \Gamma_{f_I \wedge *e_{J}} e_N, e_1 \wedge \cdots \wedge e_{N-1}) 
 = \delta_{I ,(N)} \delta_{J, (N)}=
\langle f_I\wedge *e_{J}|e_1\wedge \cdots \wedge e_{N-1} \wedge f_{N}\rangle  .
\ee
\end{proof}

\section{Factorization of  the Pl\"ucker map\\
 \hbox{$\Pl_{k,N}: \Gr_k(V) \ra \Pb(\Lambda^k V)$}}
\label{plucker_k_N_factoriz}

The  dual Pfaffian line bundles  $\Pf^*_\pm \ra \Gr_V^0(V+V^*)$ may also be related to the dual determinantal
line bundles $\Det^* \ra \Gr_k(V)$ over the Grassmannian of $k$-planes in $V$, for $k=1, \dots N-1$,
by composition of the tautological embedding of  $\Gr_k(V)$ in the isotropic Grassmannian $\Gr^0_V(V+V^*)$
\bea
\iota_{V,k}: \Gr_k(V) &\&\ra \Gr^{0\pm}_V(V+V^*) \cr
\iota_{V, k}: v &\& \mapsto v+ v^\perp \in \Gr^{0\pm}_V(V+V^*),
\label{taut_embed}
\eea
where $v^\perp \ss V^*$ is the $N-k$ dimensional annihilator of the $k$-dimensional subspace 
$v = \Span\{v_1, \dots, v_k\} \in \Gr_k(V)$,
with the Pl\"ucker map 
\bea
\Pl_{k, N}: \Gr_k(V)&\& \ra \Pb(\Lambda^k V) \cr
\Pl_{k, N}: v &\& \mapsto [ v_1 \wedge \cdots \wedge v_k ]
\eea 
embedding the Grassmannians $\Gr_k(V)$ of $k$-planes in $V$ into the projectivization of the $k$th exterior power
$\Lambda^k(V)$.

We  then have the following sequence of embeddings and pull-backs
\be
\begin{matrix}
\Det^*_k & & \Pf^*_{(-1)^k } & & \OO(1)(\Pb(\Lambda_{(-1)^k}(V)))  \cr
\downarrow & & \downarrow &  &\downarrow \cr
\Gr_k(V) & {\iota_{V,k} \atop \lra} & \Gr^0_{V_{(-1)^k}}(V+V^*) & {\Ca_V \atop \lra} & \Pb(\Lambda (V_{(-1)^k})),
\end{matrix}
\label{plucker_cartan_factoriz}
\ee
giving the dual determinantal line bundle  $\Det^*_k \ra \Gr_k(V)$ as the pull-back under $\iota_{V,k}$
of the dual Pfaffian line bundle $\Pf^*_\pm \vert_{\iota_{V,k}(\Gr_k(V))}$ over \hbox{$\Gr^{0\pm}_V(V+V^*)$}  restricted to the image 
$\iota_{V,k}(\Gr_k(V))$  of $\Gr_k(V)$ under $\iota_{V,k}$.
\begin{proposition}
\label{plucker_iota_cartan}
\be
\Det^*_k = \iota_{V,k}^*(\Pf^*_{(-1)^k} )\vert_{\iota_{V,k}(\Gr_k(V))},
\ee
and
the Pl\"ucker map 
\be
\Pl_{k, N}: \Gr_k(V) \ra \Pb(\Lambda^k(V))
\ee
factorizes through the Cartan map
\be
\Pl_{k, N} =  \Ca_V \cdot \iota_{V,k}.
\ee
\end{proposition}
\begin{proof}
This follows immediately from the definitions
\bea
\Pl_{k,N} : \Gr_k(V) &\&\ra \Pb(\Lambda^k V) \cr
\Pl_{k,N} : \Span\{v_1, \dots, v_k\}&\& \mt [v_1 \wedge \cdots\wedge v_k]
\eea
of the Pl\"ucker map, (\ref{Gamma_w}), (\ref{Cartan_map_def}) of the Cartan map $\Ca_V$ 
and (\ref{taut_embed}) of the tautological embedding $\iota_{V,k}$.
\end{proof}

 
\section{Pl\"ucker and Cartan coordinates}
\label{plucker_cartan}

 The Pl\"ucker coordinates $\{\pi_\lambda(w)\}$ of an element $w \in \Gr_V(V+ V^*)$
are, up to projectivization, the coefficients in the expansion  of the image of the Pl\"ucker map 
in the basis $\{|\lambda\rangle\}$:
\be
\Pl_V(w) = \big[\sum_{\lambda \subseteq (N)^N} \pi_\lambda(w) |\lambda\rangle\big].
\ee
 Equivalently, if we  define  $W$  to be the $2N \times N$ dimensional rectangular matrix whose $j$th column is the $j$th basis vector $w_j$ for $w\in \Gr_V(V+V^*)$, expressed as a column vector relative to the basis  $(e_1, \dots, e_N, \dots , e_{2N})$, 
  and  $W_{\lambda}$ to be the $N \times N$  submatrix whose $i$th row is the $l_i$th row of $W$, 
  we have, up to projectivization,
 \be
\pi_\lambda(w) = \det(W_\lambda). 
 \ee
 Thus, each Pl\"ucker coordinate $\pi_\lambda(w)$ is in fact a holomorphic section
 of the dual determinantal line bundle $\Det^* \ra \Gr_V(V+V^*)$, and the full set of Pl\"ucker coordinates provides a basis for the 
 space $H^0(\Gr_V(V+V^*), \Det^*)$ of holomorphic sections (see Ref.~(\cite{GH}).)
 
 Recall also that, by the (generalized) Giambelli identity (cf.~\cite{HB}, Appendix C.8), if we express $\lambda$ in Frobenius notation \cite{Mac} as
 \be
 \lambda(\alpha, \beta)=(\alpha_1, \dots, \alpha_r | \beta_1, \dots, \beta_r), 
\ee
with
\be
 N> \alpha_1 > \cdots >\alpha_r \ge 0, \quad  N> \beta_1 > \cdots >\beta_r \ge 0,
 \ee
we have
 \be
\left(\pi_\emptyset(w)\right)^{r-1} \pi_\lambda(w) = \det\left(\pi_{(\alpha_i | \beta_j)}(w)\right)_{1\le i,j \le r}.
\label{general_giambelli_id}
 \ee
 \begin{definition}
 For $\alpha = (\alpha_1, \dots , \alpha_r)$, let $\alpha'=  (\alpha'_1, \dots , \alpha'_r)$ denote the complement 
 in the $r \times N$ rectangle  (with columns labelled from $0$ to $N-1$), in reversed order:
 \be
 \alpha'_i := N-1 - \alpha_{r-i+1}, \quad 1\le i \le r.
 \ee
 \end{definition}
We denote by 
\be
\lambda(\alpha):= (\alpha_1, \dots, \alpha_r | \alpha'_1, \dots, \alpha'_r)
\ee
 the {\em pseudosymmetric}  partitions, whose Frobenius indices satisfy $(\beta=\alpha')$. 
 \begin{remark}
 If we reverse the ordering in the second half of our basis, choosing instead: $\{e_1,\dots, e_N, e_{2N}, \dots, e_{N+1}\}$,
these partitions would  correspond to $\beta = \alpha$, and hence would, in fact, be symmetric
in the usual sense.
 \end{remark}
 
 In the following, we adopt the notation $I(\alpha)$ to denote the ordered subset 
 \be
 \{I_1(\alpha), \dots, I_r(\alpha)\} \subseteq\{1, \dots, N\}
 \ee
obtained from $\alpha =(\alpha_1, \dots, \alpha_r)$, viewed as a {\em strict partition} \cite{Mac},
 by  adding $1$ to each of its parts
 \bea
 I_i(\alpha)&\&:= \alpha_i + 1, \cr
 1\le I_i(\alpha) &\& \le N, \quad i =1, \dots, r.
 \label{I_alpha_def}
 \eea
 
 Define the bases
  $\{e_{I}\}$ and $\{f_{I}\}$ for $\Lambda^r V$ and $\Lambda^r V^*$, respectively, as:
 \bea
 e_{I} &\&:= e_{I_1} \wedge \cdots \wedge e_{ I_r}\in \Lambda^r V ,
 \label{def_e_alpha}
 \\
 f_{I}&\&:= f_{I_1}\wedge \cdots \wedge f_{I_r}\in \Lambda^r V^*,
\label{def_f_alpha}
 \eea 
 where
 \be
 I = (I_1 > \cdots I_r > 0) \subseteq \{1, \dots, N\}, \quad 1\le r \le N.
 \label{I_def}
 \ee
 In particular, the volume form $\Omega_V$ on $V$ is 
 \be
  \Omega_V:=(-1)^{{N\over 2}(N-1)}e_{(N,N-1, \dots, 1)} .
 \ee 
 \begin{definition}
  Let $\tilde{\alpha}$ denote the complement of $\alpha$ in $(0, \dots, N-1)$, 
  with parts $\{\tilde{\alpha}_i\}_{i=1, \dots , N-r}$ 
 in decreasing order forming a strict partition of length $N-r$, and let
 \be
 \tilde{I}(\alpha) = I(\tilde{\alpha})
 \ee
 denote the corresponding complement of $I(\alpha)$ in $\{1, \dots ,N\}$.
 \end{definition}

 On $\Lambda^N(V+V^*)$,  we then have
 \be
 | \lambda(\alpha,  \beta) \rangle = f_{I(\alpha)} \wedge *e_{I(\beta')}.
 \label{lambda_alpha_beta_basis_el}
 \ee
 In particular, for {\em pseudosymmetric} partitions $\lambda(\alpha)$, we have
 \be
 |\lambda(\alpha)\rangle =  f_{I(\alpha)} \wedge *e_{I(\alpha)}, 
  \label{lambda_alpha_basis_el}
 \ee
which is the image, under the Pl\"ucker map, of the isotropic subspace spanned
 by the basis vectors $\{e_{\tilde{I}_1(\alpha)}, \dots , e_{\tilde{I}_{N-r}(\alpha)}, f_{I_1(\alpha)}, \dots, f_{I_r(\alpha)}\}$.
 \begin{definition}
 The sign of the strict partition $\alpha$, denoted $\sgn(\alpha)$ is
defined to be the same as that of $I(\alpha)$:
\be
\sgn(\alpha) := \sgn(I(\alpha))= (-1)^{r(N-r)} \sgn(\tilde{\alpha}).
\ee
 \end{definition}
 \begin{definition}  The {\em Cartan coordinates} $\kappa_\alpha(w^0)$ of  
 $w^0\in\Gr^0_V(V+ V^*)$
 are  defined as:
 \be
 \kappa_\alpha(w^0)  := (\Ca_V(w^0), *e_{{I}(\alpha)}) 
  =(-1)^{\tfrac{1}{2}r(r-1) + rN}\beta_0\left(\Ca_V(w^0), e_{I(\alpha)}\right).
 \label{kappa_alpha_def_invar}
 \ee
We may therefore express its Cartan image as
 \be
 \Ca_V(w^0) = \big[\sum_{\alpha \, \in \Big\{{ \text{strict partitions of parity } \pm(-1)^N, \atop \text{ where } w^0\in \Gr^0_{V_\pm}(V+V^*)} \Big\}}
\kappa_\alpha (w^0) *e_{I(\alpha)}\big].
 \label{cartan_map_exp}
 \ee
 \end{definition}

On the intersection
\be
U^0_{\lambda(\alpha)} := U_{\lambda(\alpha)} \cap \Gr^0_V(V+V^*)
\ee
of the coordinate  neighbourhood  $U^0_{\lambda(\alpha)} \ss  \Gr_V(V+ V^*)$ on which 
 \be
 \det(W_{\lambda(\alpha)}) \neq 0,
 \ee
 with $\Gr^0_V(V+V^*)$, let  $A^{\lambda(\alpha)}(w^0)$ denote the $N\times N$ submatrix of 
 $W W_{\lambda(\alpha)}^{-1}$   whose rows are in the complementary positions to those of $W_{\lambda(\alpha)}$. 
The condition that $w^0$ is isotropic  is equivalent to the fact that this is a skew symmetric matrix
 \be
 \left(A^{\lambda(\alpha)}(w^0)\right)^T= - A^{\lambda(\alpha)}(w^0).
 \ee 
In particular, for the null partition $\alpha=\emptyset$, $A^\emptyset(w^0)$ is the affine coordinate matrix 
of the element $w^0\in \Gr^0_V(V + V^*)$ on  the ``big cell'' $U^0_{\emptyset}$.

By the generalized Giambelli identity (\ref{general_giambelli_id}),  the Pl\"ucker coordinates of $w^0\in U^0_{\emptyset}$ 
may be expressed as determinants of minors of the affine coordinate matrix $A^\emptyset(w^0)$.
\begin{lemma}\label{affinePlucker} 
For strict partitions $\alpha, \beta$ of length $r$, let $A_{(I(\alpha) | I(\beta))}$  
denote the submatrix of the matrix $A^\emptyset(w^0)$  with rows $I(\alpha)$ and columns  $I(\beta)$. Then,
up to projective equivalence,  we have the following expression for the Pl\"ucker coordinates of $w^0$, 
\be 
\pi_{\lambda(\alpha | \beta')}(w^0) =\det (A_{(I(\alpha) | I(\beta))}) = (f_{I(\alpha)} \wedge * e_{I(\beta)},  \Pl_V(w^0)) .
\ee
\end{lemma}
The Cartan coordinates  $\{\kappa_\alpha(w^0)\}$  may similarly be expressed as Pfaffians of the principal minors of  $A^\emptyset(w^0)$.
\begin{definition}
For any skew symmetric $N\times N$ matrix $A$, and any strict partition $\alpha$ of length $r=\ell(\alpha)$, between $1$ and $N$, let $A(\alpha)$ denote the $r \times r$ principal minor of $A$ with  rows and columns  in the 
positions $(I_1(\alpha), \dots, I_r(\alpha))$.
\end{definition}
 As $\alpha$ varies over the $2^N$ strict partitions $\alpha$ for which $\lambda(\alpha) \subseteq (N)^N$,
 applying (\ref{kappa_alpha_def_invar}), the Cartan coordinates $\{\kappa_\alpha(w^0)\}$ 
for $w^0\in U^\emptyset$ are given by the following.
\begin{proposition}
\label{kappa_alpha_big_cell}
On the big cell $U^0_{\emptyset}$, the Cartan coordinates are
 \be
 \kappa_\alpha(w^0) =(-1)^{r\over 2} \Pf(A^\emptyset(w^0) (\alpha)),
 \label{kappa_alpha_pfaffian}
\ee 
up to projective equivalence.
\end{proposition}
\begin{proof}
For $w^0\in U^0_\emptyset$, we may choose the basis 
\be
w_i = e_i + \sum_{j=1}^N A^\emptyset_{ji}f_{j}, \quad i =1 , \dots, N,
\label{big_cell_basis}
\ee
and hence 
\be
\Gamma_{w_i}  = \psi_i + \sum_{j=1}^N A^\emptyset_{ji}\psi^\dag_{j}, \quad i =1 , \dots, N.
\ee
The homogeneous $2N \times N$ coordinate matrix representing $w^0$ in the basis $(e_1, \dots, e_N, f_1, \dots, f_N)$
  is thus
\be
\begin{pmatrix}\Ib_N \cr A^\emptyset(w^0) \end{pmatrix} = \exp\begin{pmatrix} {\bf 0} &  {\bf 0} 
 \cr A^\emptyset(w^0)&  {\bf 0}  \end{pmatrix}  \begin{pmatrix}\Ib_N \cr {\bf 0}\end{pmatrix}
= \begin{pmatrix} \Ib_N & {\bf 0}  \cr A^\emptyset(w^0) & \Ib_N \end{pmatrix} \begin{pmatrix}\Ib_N \cr {\bf 0} \end{pmatrix}.
\ee

 Eq.~(\ref{cartan_V}), together with the equivariance of the Cartan map, imply that 
\bea
\Ca(w^0)  &\&= \left[e^{\sum_{1\le i < j \le N} (A^\emptyset)_{ij}(w^0) \psi^\dag_i \psi^\dag_j} \cdot \Omega_V \right] 
\cr
&\& = \left[\sum_{n=0}^{[N/2]} {1\over n!} \left(\sum_{1\le i < j \le N} (A^{\emptyset})_{ij} \psi^\dag_i  \psi^\dag_j \right)^n \Omega_V \right]\cr
&\& \phantom{X} \cr
&\& =\left[ \sum_{\alpha \in  \{{\text{strict partitions  of even cardinality }  r}\}}   \kappa_\alpha(w^0) *e_{{I}(\alpha)}\right], 
\label{cartan_map_affine_basis}
\eea
where
\be
\kappa_\alpha(w^0) := (-1)^{r\over 2}\Pf\left (A^\emptyset(w^0)(\alpha)\right).
\label{cartan_coords_aff}
\ee
\end{proof}
\section{Coordinate interpretation: a Pfaffian analog of the Cauchy-Binet identity}
\label{pfaffian_CB}

Choose two pairs 
\be
I, J \subseteq \{1, \dots, N\}, \quad K,L\subseteq \{1, \dots, N\}
\ee
of decreasingly ordered subsets
\bea
I &\&= (I_1, \dots, I_i), \quad J= (J_1, \dots , J_j) \cr
K &\&= (K_1, \dots, K_k), \quad L= (L_1, \dots , L_l)
\eea
 with cardinalities
 \be
 |I| =i, \quad |J| =j, \quad |K| = k, \quad |L| = l.
 \ee
satisfying
\be
i+j =k+l = 2r.
\ee
for some $r$, with 
\be
1\leq r \leq N.
\ee
\br
In Theorem \ref{pfaffian_CB_identity} below, only the case where  $I$ and $J$ have  equal cardinalities
\be
i = j= r
\ee
 will be needed. But for the moment, we do not require this restriction.
 \er
  Let $\II_{(I J | K L)}$ be the $2r \times 2r$ submatrix of the $2N \times 2N$ matrix
 with block form
 \be
\II_{2N} := \begin{pmatrix} \Ib_N & \Ib_N \cr \Ib_N &  -\Ib_N \end{pmatrix},
 \ee
 whose first $i$ rows are the rows of $\II_{2N}$ in positions $\{I_i\}_{i=1, \dots , r}$ and next $j$ rows 
 are those in positions $\{J_{i+N}\}_{i=1, \dots ,r}$,
 and whose first $k$ columns are those of $\II_{2N}$ in positions $\{K_i\}_{i=1, \dots , k}$ and 
 next $l$ columns are those in positions $\{L_{i +N}\}_{i=1, \dots, l}$.
 Now define
 \be
 \Delta_{(I J | K L)} := \det(\II_{(I J | K L)}).
 \ee
 \begin{lemma}
 \label{Delta_IJKL_abs}
$\Delta_{(I J | K L)}$ vanishes unless the set theoretic conditions
\be
K\cup L=I \cup J \text{ and }  K\cap L =I \cap J
\label{IJKL_cond}
\ee
are satisfied.
 \end{lemma}
 
  \begin{proof}
  Elements of $K\cup L$ that are not in $I\cup J$ give vanishing columns of $\II_{(I J | K L)}$,
  and elements of $I\cup J$ that are not in $K\cup L$ give vanishing rows. In either case, the
  determinant $\Delta_{(I J | K L)}$ vanishes. Therefore, if it does not vanish, the first
  equality of (\ref{IJKL_cond}) must hold.  The rows and columns of $\II_{(I J | K L)}$ each contain
  either no nonzero entries, or one, equal to $\pm 1$, or two, equal to $\pm1$. The elements of $I\cap J$
  give pairs of rows, in the top and bottom half, which either vanish, or have a single element $\pm 1$ that
  is nonzero in the same place, or a pair of elements with nonzero entries $(1,1)$ and $(1,-1)$. In the first two cases, the
  determinant vanishes. If the element is not in $K\cap L$, it cannot be the third case, so we must have
  $I\cap J \subseteq K \cap L$ for the determinant not to vanish. The same consideration, interchanging rows and
  columns, shows we must have $K\cap L \subseteq I \cap J$  for it not to vanish. Therefore, if it does not vanish,
  the second equality of (\ref{IJKL_cond}) is true.
    \end{proof}
    
     \br
 \label{IJKL_interdep}
 Let
 \be
T:=I\cup J = K\cup L
\ee 
 The equalities  (\ref{IJKL_cond}) are equivalent to the fact that the following are disjoint subsets of $T$:
\bea  
S&\&:= I\cap J = K\cap L, \cr
A&\&:= (I\cap K)\backslash S, \quad B:= (I\cap L)\backslash S,  \cr
\quad C&\&:= (J\cap K)\backslash S,  \quad D:= (J\cap L)\backslash S,
\eea
whose union is $T$,
\be
T = A \sqcup B \sqcup C \sqcup D \sqcup S.
\label{disjoint_union_ABCDS}
\ee
It follows that
\bea
I&\&= A\sqcup B\sqcup S, \quad J= C\sqcup D\sqcup S,  \cr
K&\&= A\sqcup C\sqcup S, \quad  L= B\sqcup D\sqcup S.
\label{IJKL_ABC}
\eea
\er

Denote  the cardinalities of $(A,B,C,D,S,T)$
  \be
  |A|=:a, \quad |B|=: b, \quad |C|=: c, \quad  |D|=: d,\quad |S|=: s, \quad |T|=: t.
  \ee
Then
\bea
 t+s &\& = 2r = i + j, \cr
t-s &\& =2t-2r = a + b + c + d, \cr
 i&\&= a+b+s, \quad j= c+d+s, \cr
 k &\&= a+c+s, \quad
 l= b+d+s,
\label{abcdts}
\eea
and therefore  $a+b+c+d$ must be even.
\br
Note also that, for given $(I,J)$,  conditions \eqref{IJKL_cond} imply that $L$ 
is uniquely determined once $K$ is given, and vice versa:
\be
L = ((I\cup J)\setminus K)\cup(I\cap J), \quad K = ((I\cup J)\setminus L)\cup(I\cap J).
\ee
\er

\begin{remark}
In the subsequent development, up to the statement of Theorem  \ref{pfaffian_CB_identity}, 
the sets $(I,J,K,L)$ will, for simplicity of notation, always be  understood 
as being in increasing order, rather then decreasing, as has been used till now 
(in deference to the conventions used for partitions). This global reversal is purely a matter of 
notational convenience, however, that will leave invariant all relative signs  to be computed.
\end{remark}

It will be pertinent to look at classes of subsets $(I,J,K,L)$  that share the same values of $(a,b,c,d,s)$. 
These are all related by applying some suitable permutation to the case in 
which $(A,B,C,D,S)$ are the following
 \bea
 A &\&= (1, \dots, a),  \quad B = (a+1, \dots, a+b), \cr
 C &\&= (a+b+1, \dots, a+b+c),\quad D = (a+b+c+1, \dots, a+b+c+d),\cr
  T &\& = (1, \dots, t), \quad  S = (a+b+c+d+1, \dots , a+b+c+d+s=t).
  \label{ABCDS_ordering}
 \eea
Note that, if $t<N$, the subset of $(1, \dots, N)$ complementary to $T$ plays no role in the calculation, so
it will be placed, with no  loss of generality, at the end of the sequence.

Putting the sets $A,B,C,D,S$ into the order (\ref{ABCDS_ordering}), while maintaining their
internal order as increasing,  entails a change of sign in the determinant
$\Delta(IJ| KL)$ that needs to be computed. To do this,  define the following.
 \begin{definition}
For any given $(I,J, K, L)$ satisfying  (\ref{IJKL_cond}) of Lemma \ref{Delta_IJKL_abs}, define  
 \be 
\nu(I,J, K, L) :=  \mu(A,B)  + \mu(A,C) + \mu(C,D) + \mu(B, D),
\label{nu_IJKL_def}
\ee 
where, for any pair of disjoint subsets $E,F \subseteq \{1,\dots,N\}$, 
\be 
\mu(E,F) :=  \#\{(i,j)| i\in E, j\in F, i>j\}
\ee 
\end{definition}
Alternatively, let $(K^+ | L^-)$ be the ordered sequence of distinct $2r$ pairs
\be
((K_1,+), \dots, (K_k,+) ,| (L_1,-), \dots ,( L_l,-))
\ee
consisting of the parts of $K$, in order, together with a $+$ sign, followed by those of $L$, together with a $-$ sign.
Define a second sequence $(I^\epsilon | J^\epsilon)$ consisting of another distinct $2r$ pairs
\be
((I_1, \epsilon_1), \dots (I_i, \epsilon_i), (J_1, \epsilon_{i+1}), \dots (J_j,  \epsilon_{2r})),
\ee
where $\{\epsilon_m= \pm\}_{1\le m \le 2r}$, with $\epsilon_m= +$ whenever $I_m \in S$
and $\epsilon_{i+m}= -$ whenever $J_m \in S$, and all other $\epsilon_m$'s (uniquely)
determined by the requirement that the ordered sequence of  pairs $(I^\epsilon | J^\epsilon)$ 
be a permutation of the sequence $(K^+ | L^-)$.
\begin{definition}
Let $\sgn(I,J,K,L)$ be the sign of the permutation in $\SS_{2r}$ that takes  $(K^+ | L^-)$ into $(I^\epsilon | J^\epsilon)$.
\end{definition}
Assuming $k$ and $l$ to be even, these two are related by:
\begin{lemma}
\be
\sgn(I,J,K,L)= (-1)^{\nu(I,J,K,L)+b+jd}
\ee
\end{lemma}
This follows from comparing the definitions of  $\sgn(I,J,K,L)$ and $\nu(I,J,K,L)$, using
the relations (\ref{abcdts}), the fact that $k$ and $l$ are even and $m^2 \equiv m \  \mod \, 2$,
which together imply
\be
 \pm b \pm c \equiv \pm a\pm d \ \ \mod \, 2.
\ee
The following gives the complete expression for $\Delta(IJ| KL)$.
  \begin{lemma}
 \label{Delta_IJKL_result }
 \be
 \Delta_{(IJ | KL)}=  (-1)^{jd + \nu(I,J,K,L)}  2^s \delta_{I\cup J, K\cup L} \delta_{I\cap J, K\cap L}.
 \label{Delta_value_sign_jd}
\ee
 \end{lemma}
 \begin{proof}
The idea is to reduce to the case when the basis elements are as in \eqref{ABCDS_ordering}, and  compute the change of sign.
Assume first that the entries of the matrix $\II_{(I J | K L)}$ are chosen so that $(A,B,C,D,S)$
are as in (\ref{ABCDS_ordering}), and denote this as $\II^0_{(I J | K L)}$.
Then $\II^0_{(IJ | KL)}$ has the  block form
 \be
 \II^0_{(IJ | KL)} = \begin{pmatrix} \Ib_a &\Ob_{ac}& \Ob_{as} & \Ob_{ab} & \Ob_{ad} &\Ob_{as}\cr
   \0b_{ba} &\Ob_{bc}& \Ob_{bs} & \Ib_{b} & \Ob_{bd} &\Ob_{bs}\cr
    \0b_{sa} &\Ob_{sc}& \Ib_{s} & \0b_{sb} & \0b_{sd} &\Ib_{s}\cr
     \0b_{ca} &\Ib_{c}& \Ob_{cs} & \0b_{cb} & \0b_{cd} &\0b_{cs}\cr
 \0b_{da} &\Ob_{dc}& \Ob_{ds} & \0b_{db} & -\Ib_{d} &\0b_{ds}\cr
 \0b_{sa} &\Ob_{sc}& \Ib_{s} & \0b_{sb} & \0b_{sd} &-\Ib_{s},
 \end{pmatrix},
 \ee
 where $\Ib_n$ denotes the $n\times n$ identity matrix and $\Ob_{mn}$ denotes the $m\times n$ matrix with vanishing entries.
By successively eliminating unit matrix blocks, the determinant therefore reduces in this case  to
\bea
\Delta^0_{(IJ | KL)} &\& := \det(\II^0_{(IJ|KL)}) 
 = (-1)^{(b+s)c+ bs +d} \det \begin{pmatrix} \Ib_s & \Ib_s \cr \Ib_s &- \Ib_s \end{pmatrix} \cr
 &\&  = (-1)^{(b+s)c+ bs +d+s}  2^s.
\eea

Now let $\sigma$ be the permutation of shuffle type that puts $A, B,C,D,S$  in the order (\ref{ABCDS_ordering}), preserving the internal order of each of  $A, B,C,D,S$.  On the first $r$ rows of $\II_{(I J | K L)}$, indexed by $I$,  this puts the elements of $I$ in the order $A,B,S$ by sliding the elements of $A$ past $B$ and $S$, and  the elements of $B$ past the elements of $S$.   The sign of the determinant therefore changes by $(-1)^{\mu(A,B) + \mu(A,S) + \mu(B,S)}$. Similarly, moving the elements of the second set of rows indexed by $J$ into the order $C,D,S$ changes signs by $(-1)^{\mu(A,B) + \mu(A,S) + \mu(B,S)}$. Moving the columns of the matrix $\II_{(I J | K L)}$ indexed by $K,L$ into the order induced by (\ref{ABCDS_ordering}) gives a sign change $(-1)^{\mu(A,C) + \mu(A,S) + \mu(C,S)+\mu(B,D) + \mu(B,S) + \mu(D,S)}$. Composing all these signs therfore gives a total sign change $(-1)^{\nu(I,J,K,L)}$
so that
\be
\Delta(IJ|KL) = (-1)^{\nu(I,J,K,L)  + (b+s)c+ bs +d+s}2^s \delta_{I\cup J, K\cup L} \delta_{I\cap J, K\cap L}.
\label{Delta_value_sign}
\ee
To obtain (\ref{Delta_value_sign_jd}), since $l$ is even,  substitute   
\be
c = j -d -s, \quad b+s \equiv d \, (\mod\,2) 
\ee
in the exponent of $-1$ in (\ref{Delta_value_sign}) and reduce  $\mod\,2$. It then follows that for general choice
of $(I,J,K,L)$ the determinant $\Delta(IJ | KL)$ is given by (\ref{Delta_value_sign_jd}).
 \end{proof}
Henceforth, we also assume  that $i=j=r$, so eq. (\ref{Delta_value_sign_jd}) becomes
\bea
 \Delta_{(IJ | KL)} &\&=  (-1)^{rd + \nu(I,J,K,L)}  2^s \delta_{I\cup J, K\cup L} \delta_{I\cap J, K\cap L} \cr
&\& =  (-1)^b 2^s\sgn(I,J,K,L) \delta_{I\cup J, K\cup L} \delta_{I\cap J, K\cap L}.
 \label{Delta_value_sign_rd}
\eea
Now let $A$ be a skew symmetric $N \times N$ matrix and, for any pair $K, L \subseteq \{1, \dots, N\}$
of ordered subsets of cardinalities $k, l$, respectively,  with $k+l=2r$, let  $A_{(K | L)}$ denote the $k \times l $ 
submatrix of $A$ whose rows and columns are the restriction of those of $A$ in positions $K$ and $L$, respectively.

For $w^0\in \Gr^0_V(V+V^*)$ in the big cell, choose $A$ to  equal the affine coordinate matrix   
\be
  A:=A^\emptyset(w^0).
  \ee
Theorem \ref{main_theorem}  is then  equivalent to the following identity.
 \begin{theorem}
 \label{pfaffian_CB_identity}
 \be
 \det(A_{(I | J)}) =  {(-1)^{\tfrac{1}{2}r(r-1)}\over 2^r }\sum_{K, L | k, l \text{ even} \atop{ {K\cup L = I \cup J \atop K\cap L = I\cap J}}}(-1)^{{l\over 2}  }\Delta_{(I J | K L)}
 \Pf(A_{(K | K)})\Pf( A_{(L | L)}).
 \label{pfaff_CB_ident}
 \ee
 \end{theorem} 
 \br
 Note that the increasing order within the sets (I,J,K,L) may be simultaneously reversed without changing
 the equality (\ref{pfaff_CB_ident}). This must done to revert to the decreasing order
 used in Section \ref{plucker_cartan}. Then  $(I , J)$, written in decreasing order, is
 related to the strict partitions $(\alpha, \beta)$ that determine the partition $\lambda(\alpha, \beta)$
 by
 \be
 I =I(\alpha), \quad J=I(\beta),
 \label{IJ_alpha_beta}
 \ee
and hence have equal cardinality $|I|=|J| = r$, while $(K, L)$ are related to the strict 
partitions $(\gamma, \delta)$ that determine the pseudosymmetric partitions 
$\lambda(\gamma)$ and $\lambda(\delta )$ by
\be
K =I(\gamma), \quad
L =I(\delta),
\label{KL_gamma_delta}
\ee
 so
\be
\det(A_{(I |J)})=\pi_{\lambda(\alpha,\beta')}(w), \quad \Pf(A_{(K |K)}) = \kappa_\gamma(w), \quad \Pf(A_{(L |L)}) = \kappa_\delta(w).
\ee
For $I=J$ the only admissible pair in the above sum is $(K,L)=(I,I)$, and therefore \eqref{pfaff_CB_ident} reduces 
to the standard identity
\be
\det(A_{(I | I)}) =  
\begin{cases}
\Pf(A_{(I | I)})^2 & \text{ if } r \text{ is even}\\
0 & \text{ if } r \text{ is odd.}
\end{cases}
\ee
\er

In view of the fact that, by  (\ref{IJKL_ABC}),  $I$, $J$ and $K$ uniquely determine $L$, 
and $I$, $J$ and  $L$ uniquely determine $K$, eq.~\eqref{pfaff_CB_ident} really consists only 
of a single sum, either over the variable $K$ or $L$ which, using eq.~(\ref{Delta_value_sign_rd}) gives:
\begin{corollary}
\label{single_sum_det_pfaff_ident}
\bea
 \det(A_{(I | J)})  &\&=  {(-1)^{\frac{1}{2}r(r-1)}\over  2^{r-s}}\sum_{K, L | k, l \text{ even} \atop{ {K\cup L 
 = I \cup J \atop K\cap L = I\cap J}}}(-1)^{{l\over 2}+rd+ \nu(I,J, K, L)}
 \Pf(A_{(K | K)})\Pf( A_{(L | L)})\cr
&\& = {(-1)^{\frac{1}{2}r(r-1)}\over  2^{r-s}}\sum_{K, L | k, l \text{ even} \atop{ {K\cup L = I \cup J \atop K\cap L = I\cap J}}}(-1)^{{l \over 2}+b }\sgn(I,J,K,L)
 \Pf(A_{(K | K)})\Pf( A_{(L | L)}).\cr
 &\&
 \label{single_sum_pfaff_CB_ident}
 \eea
\end{corollary}

Before proceeding to a direct proof of eq.~(\ref{pfaff_CB_ident})  we show that,
under the identifications (\ref{IJ_alpha_beta}), (\ref{KL_gamma_delta}),
it  is just the coordinate expression of eq.~(\ref{beta_sigma_Ca_Ca_dual}), and hence,
up to projectivization, is equivalent to Theorem \ref{main_theorem} expressed in affine coordinates.

\begin{proposition} 
\label{equiv_main_theorem_pfaff_CB_ident}
Under the identifications (\ref{IJ_alpha_beta}), (\ref{KL_gamma_delta}),
eqs.~(\ref{beta_sigma_Ca_Ca_dual}) and (\ref{pfaff_CB_ident}) are equivalent (up to projectivization).
\end{proposition}
\begin{proof}
As shown in Section \ref{proof_main_theorem},  eq.~(\ref{beta_sigma_Ca_Ca_dual}) is equivalent to  verifying
\be
(\Gamma_{f_I\wedge *e_{J}} C(\Ca_V(w^0)), \Ca_V(w^0))) = 
\langle f_I \wedge *e_{J} | \Pl_V(w^0)\rangle .
\label{gamma_f_I_wedge_e_star_J_Ca_diag_bis}
\ee
for all $\sigma$ of the form
\be
\sigma= f_I \wedge *e_{{J}},
\label{sigma_f_I_star_e_J_bis}  
\ee
Lemma (\ref{affinePlucker}) gives 
\be
 \langle f_I \wedge *e_{J} | \Pl_V(w^0)\rangle =(-1)^{r(N-r)}\NN_N \, \det(A_{(I | J)}),
 \label{plucker_det_A_IJ_new}
 \ee
 where $\NN_N\neq 0$ is any projective normalization factor. Making the identifications (\ref{IJ_alpha_beta}), (\ref{KL_gamma_delta}), 
 substituting (\ref{plucker_det_A_IJ_new}) into the RHS of eq.~(\ref{gamma_f_I_wedge_e_star_J_Ca_diag_bis}), 
and (\ref{cartan_map_exp}) twice into the LHS, 
using  eq.~(\ref{kappa_alpha_pfaffian}) to express the Cartan coordinates as Pfaffians,  
and (\ref{C_Hodge_sign}), (\ref{C_squared}), to relate the map $C$ to Hodge $*$, and using
the fact that $k$ and $l$ are even gives:
\be
 \det(A_{(I | J)}) ={(-1)^{{1\over 2}N(N+1)} \over \NN_N}\sum_{K, L | k, l \text{ even} \atop{ {K\cup L = I \cup J \atop K\cap L = I\cap J}}}
 (-1)^{l \over 2} \widehat{\Delta}_{(I J | K L)}  \Pf(A_{(K | K)})\Pf( A_{(L | L)}) ,
\label{tilde_pfaff_CB_ident_hat}
\ee
where
\be
\widehat{\Delta}_{(I J | K L)} :=(\Gamma_{f_I\wedge *e_{{J}}}e_{K}, *e_{{L}}).
\label{tilde_delta_IJKL}
\ee
Lemma \ref{hat_Delta_Delta_rel}  below shows that this is equivalent to:
\be
\det(A_{(I | J)}) = {(-1)^{{1\over 2}N(N+1)} \over \NN_N}\sum_{K, L | k, l \text{ even} \atop{ {K\cup L = I \cup J \atop K\cap L = I\cap J}}}  
(-1)^{{l\over 2 } }{\Delta}_{(I J | K L)}  \Pf(A_{(K | K)})\Pf( A_{(L | L)}). 
\label{tilde_pfaff_CB_ident}
\ee
To  obtain the correct normalization, it is sufficient to evaluate this for one specific choice of $(I | J)$. But for any $I=J$, we have
\be
 \det(A_{(I | I)}) = \Pf(A(I|I))^2,
 \ee
so 
 \be
 \NN_N = (-1)^{{1\over 2}N(N+1)},
\ee
and we obtain eq.~(\ref{pfaff_CB_ident}).
\end{proof}
\begin{lemma}
\label{hat_Delta_Delta_rel}
Let $I$ and $J$ have the same cardinality $i=j=r$, and let $K$, and hence also $L$, have even cardinalities $(k, l)$.
Then  $\widehat{\Delta}_{(I J | K L)}$ is nonzero if and only if eq.~(\ref{IJKL_cond}) holds and, in that case,
\be
  \widehat{\Delta}_{(I J | K L)} =  \frac{ (-1)^{ {1\over 2} r(r-1) }}{2^{ r} }\Delta_{(I J | K L)}.
  \label{hat_Delta_Delta}
\ee
  \end{lemma}
  \begin{proof}
First note that, if $i \neq j$, the product  $\Gamma_{f_i} \Gamma_{e_j}$  acting on $e_K$  gives a multiple of a homogeneous basis element
\be
 \Gamma_{f_i} \Gamma_{e_j} e_K =  -\Gamma_{e_j} \Gamma_{f_i} e_K =  \kappa e_{M}, \quad \kappa = \pm 1 \, \text{ or } 0
 \ee
  for some $M$ of the same cardinality as $K$  If $i=j$, either  $\Gamma_{f_i} \Gamma_{e_j} e_K$  or 
  $\Gamma_{e_j} \Gamma_{f_i} e_K $ vanishes, and the other equals $e_K$.
       It follows that $\Gamma_{f_I\wedge e_{\tilde{J}}}e_{K}$ is a monomial, and hence a multiple 
   \be
  \Gamma_{f_I\wedge *e_{J}}e_{K}= \gamma \cdot e_{P}
  \ee
   for some scalar $\gamma$ and some  ordered subset $P\subseteq\{1, \dots, N\}$. 
   
  If   $\hat{\Delta}_{(I J | K L )} \neq 0$, we must therefore have $\gamma \neq 0$ and $P=\tilde{L}$, so 
\be
  \Gamma_{f_I\wedge *e_{J}}e_{K}  = \sgn(J) \Gamma_{f_I\wedge e_{\tilde{J}}}e_{K} = \pm e_{\tilde{L}}.
  \label{Gamma_IJKL}
\ee
   We next show that this  is equivalent to the set theoretic equalities (\ref{IJKL_cond}).
First, if  $l\in K$, then either $l \in I$ (and $e_l$ is removed by $\Gamma_{f_I}$) or $l \in J$ 
(since otherwise, if $l\not\in I$,  $e_l$ would  be annihilated by $\Gamma_{e_{\tilde{J}}}$). Therefore  $K\subseteq I\cup J$.
 It also follows from (\ref{Gamma_IJKL}) that if $\tilde{l}\not\in I$, and $\tilde{l}\not\in J$, then 
 $\tilde{l} \in \tilde{L}$. Therefore $\tilde{I}\cap \tilde{J}\subseteq \tilde{L}$, $L\subseteq I\cup J$ and hence $K \cup L \subseteq I \cup J$. Conversely, 
 from (\ref{Gamma_IJKL}),   if $\tilde{l}\in \tilde{L}$ and $\tilde{l}\not\in K$  then  $\tilde{l}\in \tilde{I} \cap \tilde{J}$.
Therefore $  \tilde{K}\cap\tilde{L} \subseteq \tilde{I}\cap \tilde{J} $ and hence  $I \cup J \subseteq  K \cup L$.
Combining gives $I\cup J = K\cup L$. The proof that $I\cap J = K\cap L$ follows similarly.

This means that we can decompose the subsets $I,J,K,L$ in the same way as for the case ${\Delta}_{(I J | K L)}$ above. 
Let us first consider the case of $A,B,C,D,S$ chosen as in (\ref{ABCDS_ordering}).  
For this case, we denote the quantity $\hat{\Delta}(IJ | KL)$ as  $\hat{\Delta}^0(IJ | KL)$. If 
\be
M=\{m_1<m_2<...<m_s\},
\ee 
we have
\be
\Gamma_{f_M} = \Gamma_{f_{m_1}}\Gamma_{f_{m_2}}...\Gamma_{f_{m_s}}, \quad \Gamma_{e_M} = \Gamma_{e_{m_1}}\Gamma_{e_{m_2}}...\Gamma_{e_{m_s}},
\ee
 so there is just one product. 
 
Changing $*e_J, *e_L$ to $e_{\tilde{J}} , e_{\tilde{L}}$, we  first consider 
\be
\widetilde{\Delta}^0_{(I J | K L)} :=(\Gamma_{f_I\wedge e_{\tilde{J}}}e_{K}, e_{\tilde{L}}) = \sgn(J) \sgn(L) \hat{\Delta}^0_{(I J | K L)} 
\ee
Applying $\Gamma_{f_I\wedge e_{\tilde{J}} }$  to $e_K$ gives the same result as the antisymmetrization of 
 \be
  \Gamma_{f_A} \Gamma_{f_B} \Gamma_{f_S} \Gamma_{e_A} \Gamma_{e_B} \Gamma_{e_{\tilde T}}
  \ee
applied to $e_K$.
Now note that antisymmetrization gives
\bea
\AA( \Gamma_{f_A} \Gamma_{f_B}  \Gamma_{f_S} \Gamma_{e_A} \Gamma_{e_B} \Gamma_{e_{\tilde T}}) 
&\&=(-1)^{(s+a)b}\AA( \Gamma_{f_A}  \Gamma_{f_S} \Gamma_{e_A} \Gamma_{f_B} \Gamma_{e_B} \Gamma_{e_{\tilde T}})\cr
&\&=(-1)^{(s+a)(b+a)}\AA( \Gamma_{f_S} \Gamma_{e_A} \Gamma_{f_A}  \Gamma_{f_B} \Gamma_{e_B} \Gamma_{e_{\tilde T}})
\eea
 and apply $ \Gamma_{f_S} \Gamma_{e_A} \Gamma_{f_A}  \Gamma_{f_B} \Gamma_{e_B} \Gamma_{e_{\tilde T}}$  to $e_K = e_A\wedge e_C\wedge e_S$:
  \bea 
&\& (-1)^{(s+a)(b+a)} \Gamma_{f_S}  \Gamma_{e_A} \Gamma_{f_A}  \Gamma_{f_B} \Gamma_{e_B} \Gamma_{e_{\tilde T}}(e_A\wedge e_C\wedge e_S)\cr
&\&=(-1)^{(s+a)(b+a) + {\tilde t}(a+c+s)} \Gamma_{f_S} \Gamma_{e_A} \Gamma_{f_A}  \Gamma_{f_B} \Gamma_{e_B} (e_A\wedge e_C\wedge e_S\wedge e_{\tilde T})\cr
&\&=(-1)^{(s+a)(b+a) + {\tilde t}(a+c+s)+  {1\over 2}a(a-1) + {1\over 2}b(b-1)} \Gamma_{f_S} (e_A\wedge e_C\wedge e_S\wedge e_{\tilde T})\cr
&\&=(-1)^{(s+a)(b+a) + {\tilde t}(a+c+s)+  {1\over 2}a(a-1) + {1\over 2}b(b-1) +(a+c)s + s(s-1)/2} (e_A\wedge e_C \wedge e_{\tilde T})\cr
&\&=(-1)^{(s+a)(b+a) + {\tilde t}(a+c+s)+  {1\over 2}a(a-1) + b(b-1)/2 +(a+c)s + {1\over 2}s(s-1)} (e_{\tilde L})
 \eea
Antisymmetrizing this would  give the same result, except that we can only retain permutations in which the factors  $ \Gamma_{e_A}$
  precede the factors $ \Gamma_{f_A}$ , and the factors $ \Gamma_{f_B}$ precede the factors $ \Gamma_{e_B}$.
  This reduces the sum by a total factor of $2^{-(a+b)}$, giving
\be
\widetilde{\Delta}^0_{(I J | K L)} ={(-1)^{(s+a)(b+a) + {\tilde t}(a+c+s)+  {1\over 2}a(a-1) + {1\over 2}b(b-1) +(a+c)s + {1\over 2}s(s-1)}\over 2^{ (a+b)} }.
\ee
The quotient is therefore
\be
\frac{\widetilde{\Delta}^0_{(I J | K L)} }{{\Delta}^0_{(I J | K L)} } ={(-1)^{(s+a)(b+a) + \tilde t(a+c+s)+  {1\over 2}a(a-1)+{1\over 2} b(b-1)
 +(a+c)s + {1\over 2}s(s-1) +(b+s)c+ bs +d+s}\over 2^{ (a+b)+s} }.
\ee
 It remains only to simplify the sign. Using $ m^2\equiv m \ (\mod\, 2)$,  the fact that 
 \be
k= a+c+s, \quad l = b+d+s 
\ee
 are  both even  and $r =a+b+s $,  implies
 \bea
&\&(-1)^{(s+a)(b+a) +  {1\over 2}a(a-1) + {1\over 2} b(b-1) +(a+c)s + {1\over 2}s(s-1) +(b+s)c+ bs +b}\cr
&\& = (-1)^{a+ab+b +bc +    {1\over 2}a(a-1) + {1\over 2}b(b-1)  + {1\over 2}s(s-1)},
\eea
 and 
 \be
(-1)^{a(a-1)/2 + b(b-1)/2 +  {1\over 2}s(s-1)}= (-1)^{{1\over 2}r(r-1) + ab + as + bs},
 \ee
so the sign becomes
\be
\frac{\widetilde{\Delta}^0_{(I J | K L)} }{{\Delta}^0_{(I J | K L)} } ={(-1)^{ (r-s) +(r-s)s +bc +{1\over 2} r(r-1)} \over 2^{ r} }.
\label{tilde_Delta_Delta_rat}
\ee

Replacement of $ e_{\tilde J} $ by $*e_J$ involves multiplication by 
\be
\sgn(J)  = (-1)^{r(r-s)},
\ee
 and replacement of  $e_{\tilde L}$ by $*e_L$ introduces the further sign change
\be
\sgn(L) = (-1)^{bc}.
\ee
Combining with  (\ref{tilde_Delta_Delta_rat}), we obtain 
\be
(-1)^{ (r+s+1)(r-s) +{1\over 2} r(r-1)}= (-1)^{ (r+s+1)(r+s) + {1\over 2}r(r-1) }= (-1)^{  {1\over 2} r(r-1) },
\ee
and therefore
\be
\frac{\widehat{\Delta}^0_{(I J | K L)} }{{\Delta}^0_{(I J | K L)} } = \frac{  (-1)^{  {1\over 2} r(r-1) }}{2^{ r} }.
\label{Deltahat_Delta_rat_go}
\ee
Now suppose that $A,B,C,D,S$ are internally ordered as increasing, but not necessarily
chosen as in  (\ref{ABCDS_ordering}). Let $\sigma$ be the permutation of shuffle type that takes $A, B,C,D,S$ and puts them in the order (\ref{ABCDS_ordering}), preserving the internal order of each. The expression giving  ${\widehat{\Delta}_{(I J | K L)} }$ cannot depend on how the elements are named  provided that, in the elements $f_I, e_J, e_K, f_L$, the basis elements are kept in the order induced by $\sigma$. These elements  must, however, appear in the natural order $\{1,\dots,N\}$. Thus in $f_I$, we must slide the elements of $A$ past $B$ and $S$, and 
the elements of $B$ past the elements of $S$. The sign of $f_I$ therefore changes by $(-1)^{\mu(A,B) + \mu(A,S) + \mu(B,S)}$,
that of $e_J$ by $(-1)^{\mu(C,D) + \mu(C,S) + \mu(D,S)}$,
that of $e_K$ by $(-1)^{\mu(A,C) + \mu(A,S) + \mu(C,S)}$ and that of $f_L$ by $(-1)^{\mu(B,D) + \mu(B,S) + \mu(D,S)}$.
(The volume element defining $*$ can also change sign, but this appears twice in the formula, contributing no net sign change.)
The global sign change is therefore again $(-1)^{\nu(I,J, K,L)}$, as defined in \eqref{nu_IJKL_def}. In other words, 
$\widehat{\Delta}_{(I J | K L)} ,{\Delta}_{(I J | K L)} $ change sign in exactly the same way under the passage to the order (\ref{ABCDS_ordering}), and so
\be
\frac{\widehat{\Delta}_{(I J | K L)} }{{\Delta}_{(I J | K L)} } = \frac{  (-1)^{  {1\over 2} r(r-1) }}{2^{ r} }.
\label{Deltahat_Delta_rat }
\ee
 \end{proof}

As preparation for the direct proof of Theorem \ref{pfaffian_CB_identity},  
we introduce a second set of basis vectors for $V+V^{*}$
\be
g_i := e_i+ f_i, \quad h_i := e_i - f_i \qquad i=1,\dots, N,
\ee
and define the $2N$ component row vectors 
\be
({\bf e}, {\bf f}) =
(e_1, e_2, \hdots e_N,   f_1, f_2,  \hdots  f_N),
\quad
( {\bf g} , {\bf h}) =
(g_1,  g_2,  \hdots  g_N,  h_1,  h_2, \hdots  h_N),
\ee
whose entries are the basis elements $\{e_i, f_i\}_{1=1,\dots, N}$ and
$\{g_i, h_i\}_{1=1, \dots, N}$, respectively.
These are related by
\bea
({\bf g} ,  {\bf h}) &\&=({\bf e}, {\bf f})\  \II_{2N},\cr
({\bf e} ,  {\bf f}) &\&=\tfrac{1}{2}({\bf g}, {\bf h})\  \II_{2N}.
\label{gh_ef}
\eea

Let  $g_I$,  $h_{J}$ $e_K$ and $f_L$  denote  the exterior forms in $\Lambda (V+ V^*)$ 
associated to ordered subsets $I, J, K$ and $L$ of $\{1,2,\dots, N\}$, defined by: 
\bea
g_I &\&:= g_{I_1}\wedge \cdots \wedge g_{I_r}, \quad h_J := h_{J_1}\wedge \cdots \wedge h_{J_r}, \cr
e_K &\&:= e_{K_1}\wedge \cdots \wedge e_{K_k}, \quad f_L := f_{L_1}\wedge \cdots \wedge h_{L_l},
\label{g_I_h_J}
\eea
By the usual formulae for changes of bases, we then have
\begin{lemma}
\label{IJKL_cond_lemma}
 For all ordered subsets $I,J \subseteq \{1,2,\dots, N\}$ whose cardinalities add up to $2r$, we have
\be
g_I \wedge h_J =  \sum_{{K, L \atop {K\cup L = I \cup J \atop K\cap L = I\cap J}} } \Delta_{(I J | K L)} \, e_K \wedge f_L,
\label{g_I_h_J_exp}
\ee
and for the dual basis,
\be
(g_I \wedge h_J)^* = {1\over 2^{2r}} \sum_{{K, L \atop {K\cup L = I \cup J \atop K\cap L = I\cap J}} } \Delta_{(I J | K L)} \, (e_K \wedge f_L)^*.
\label{g_I_h_J_exp_dual}
\ee
\end{lemma}
For (\ref{g_I_h_J_exp}) this is simply the change of basis  induced by the matrix $\II_{2N}$, and for the duals (\ref{g_I_h_J_exp_dual}), its inverse transpose which is a half of $\II_{2N}$. We now proceed to the proof of Theorem \ref{pfaffian_CB_identity}.
\begin{proof}[Proof of Theorem  \ref{pfaffian_CB_identity}]
Define the 2-form
\be
\omega=:\sum_{k,l=1}^{N}A_{kl}g_k\wedge h_l.
\ee
The skew symmetry of $A$ implies
\be
\omega= 2\sum_{1\leq k<l\leq N}A_{kl}\left(e_k\wedge e_l-f_k\wedge f_l\right) = 2(\omega_1-\omega_2),
\ee
where
\be
\omega_1 := \sum_{1\leq i<j\leq N}A_{ij}e_i\wedge e_j,\qquad \omega_2 := \sum_{1\leq i<j\leq N}A_{ij}f_i\wedge f_j.
\ee
For any $1\leq r \leq N$, the $r$th wedge power of $\omega$ can be written as
\begin{eqnarray}
\omega^{\wedge r}
&\&=\left(\sum_{i,j=1}^{N}A_{ij}g_i\wedge h_j\right)^{\wedge r}\cr
&\&=\sum_{i_1,j_1,\dots, i_r,j_r}\prod_{k=1}^{r}A_{i_k j_k}g_{i_1}\wedge h_{j_1}\wedge \cdots \wedge g_{i_r}\wedge h_{j_r}\cr
&\&=(-1)^{{1\over 2}r(r-1)}\sum_{i_1,j_1,\dots, i_r,j_r}\prod_{k=1}^{r}A_{i_k j_k}g_{i_1}\wedge \cdots \wedge g_{i_r} \wedge
 h_{j_1}\wedge \cdots \wedge h_{j_r}\cr
&\&=(-1)^{{1\over 2}r(r-1)}\sum_{I,J}\left(\sum_{\pi,\rho \in S_r}\mathrm{sgn}(\pi)\mathrm{sgn}(\rho) \prod_{k=1}^{r}A_{i_{\pi(k)}j_{\rho(k)}}\right)g_{I} \wedge h_{J}\cr
&\&=(-1)^{{1\over 2}r(r-1)}r!\sum_{I,J}\det(A_{(I |J)})g_{I} \wedge h_{J},
\end{eqnarray}
where the sign factor on the third line is obtained by moving the $h_j$ factors to the right, and the sums
 in the fourth and fifth lines are over all pairs of subsets $(I,J)$ of $\{1,2,\dots, N\}$ of cardinality $r$.

On the other hand, since 2-forms commute, the binomial formula implies 
\begin{eqnarray}
\omega^{\wedge r}
&\&=2^r\left(\omega_1-\omega_2\right)^{\wedge r}\cr
&\&=2^r\sum_{m=0}^{r}(-1)^{r-m} \binom{r}{m}\omega_1^{\wedge m} \wedge \omega_2^{\wedge (r-m)}\cr
&\&=2^r r!\sum_{m=0}^{r}(-1)^{r-m} \frac{\omega_1^{\wedge m}}{m!} \wedge \frac{\omega_2^{\wedge (r-m)}}{(r-m)!}\cr
&\&=2^r r!\sum_{m=0}^{r} \left(\sum_{K\atop |K|=2m}\Pf(A_{(K | K)})e_K\right) \wedge \left(\sum_{L\atop |L|=2(r-m)}(-1)^{|L|/2}\Pf(A_{(L | L)})f_L\right)\cr
&\&=2^r r!\sum_{K, L \atop{ k, l \text{ even} }}(-1)^{l \over 2}\Pf(A_{(K | K)})\Pf(A_{(L | L)})e_K\wedge f_L,
\end{eqnarray}
where the last sum is over all pairs of subsets $(K,L)$ of $\{1,2,\dots, N\}$ of even cardinalities $k$ and $l$
with
\be
k+l = 2r.
\ee
We therefore have the identity
\be
\det(A_{(I | J)})=(-1)^{{1\over 2}r(r-1)}2^r \sum_{K, L \atop{ k, l \text{ even} }}(-1)^{l\over 2}\Pf(A_{(K | K)})\Pf(A_{( L| L)}) 
(g_I\wedge h_J)^* \lrcorner (e_K\wedge f_L)
\ee
and Lemma \ref{IJKL_cond_lemma} implies eq.~\eqref{pfaff_CB_ident}.
\end{proof}

\br
\label{null_case}
If the cardinalities of $I$ and $J$ are different, the same calculation that leads to (\ref{pfaff_CB_ident}) yields the 
vanishing quadratic relations
\be
\sum_{K, L | k, l \text{ even} \atop{ {K\cup L = I \cup J \atop K\cap L = I\cap J}}}(-1)^{{l\over 2}+ b}\sgn(I,J,K,L)
 \Pf(A_{(K | K)})\Pf( A_{(L | L)}) =0,
 \ee
which are satisfied by the Pfaffians of principal minors of any skew $N\times N$ matrix $A$, valid for any pair $(I,J)$ of different cardinality.
In fact, this is nothing but another way to express the Cartan relations (\ref{BKP_cartan_equations}).
\er

\br
{\bf Function theoretic realizations.}
An interesting function theoretic realization of the bilinear relation between determinants of
minors of skew matrices and Pfaffians arises in terms of Riemann $\theta$ functions, using the space of second order  $\theta$ functions on the Prym variety of hyperelliptic curves as a model for the irreducible Clifford module \cite{VG}.  
Another instance, derived in \cite{HOr1} using fermionic methods, expresses  Schur functions as sums over products of Schur $Q$ functions.
 A large class of such identities, relating lattices of KP $\tau$-functions labelled by integer partitions to lattices of BKP $\tau$-functions, 
 is derived in \cite{HOr2}. 
 \er

 \bigskip
\noindent 
\small{ {\it Acknowledgements.} The authors would like to thank T.~Dinis da Fonseca and A.~Yu.~Orlov for helpful discussions.
This work was partially supported by the Natural Sciences and Engineering Research Council of Canada (NSERC). 

 \bigskip
\noindent 
\small{ {\it Data sharing.} 
Data sharing is not applicable to this article since no new data were created or analyzed in this study.
\bigskip

 
 \newcommand{\arxiv}[1]{\href{http://arxiv.org/abs/#1}{arXiv:{#1}}}


\begin{thebibliography}{99}
 

  \bibitem{Ca} E. Cartan, {\em The Theory of Spinors}, Dover Publications Inc, Mineola N.Y., 1981.
  
 \bibitem{Ch}   C. Chevalley, {\em The Algebraic Theory of Spinors and Clifford Algebras},
  Springer Verlag, Berlin, Heidelberg, 1997.

  \bibitem{DJKM1} E. Date, M. Jimbo, M. Kashiwara and T. Miwa, ``Transformation groups for soliton equations IV.
A new hierarchy of soliton equations of KP type'', {\em Physica } {\bf 4D} 343-365 (1982).
  
 \bibitem{DJKM2} E.~Date, M.~Jimbo, M.~Kashiwara and T.~Miwa, ``Transformation groups for soliton
equations'',  In: {\em Nonlinear integrable systems - classical theory and quantum theory}, 
World Scientific (Singapore),  eds. M. Jimbo and T. Miwa (1983).

 \bibitem{Dick} L.~A.~Dickey, {\em Soliton Equations and Hamiltonians Systems},  2nd ed.
Advanced Series in Mathematical Physics {\bf 26},  World Scientific (New Jersey, London, Singapore) (2003).

\bibitem{Fu} W.~Fulton, {\em Young Tableaux}, London Math. Soc. Student Texts {\bf 35}, 
Cambridge University Press, Cambridge U.K., (1997).

  \bibitem{HB}   J.~Harnad and F.~Balogh, {\em Tau Functions and their Applications}, Monographs on Mathematical Physics series,  Cambridge University Press, Cambridge, UK  (2021).
 
 \bibitem{HOr1} J.~Harnad and A.~Yu.~Orlov, ``Fermionic approach to bilinear expansions of Schur functions in Schur $Q$-functions'', arXiv:2008.13734.

 \bibitem{HOr2} J.~Harnad and A.~Yu.~Orlov, ``Bilinear expansions of lattices of KP $\tau$-functions in BKP $\tau$-functions: a fermionic approach'',  {\em J. Math. Phys.} (in press, 2021), arXiv:2010.05055.
  
 \bibitem{GH} P.~Griffiths and J.~Harris, {\em Principles of Algebraic Geometry},  Chapt. I.5
 Wiley-Interscience, John Wiley and Sons, New York, 1978. 
 
  \bibitem{HS1} J.~Harnad and S.~Shnider,  ``Isotropic geometry and twistors in higher dimensions 
  I. The generalized Klein correspondence and spinor flags in even dimensions,''  {\em J. Math. Phys.} {\bf 33}, 3191-3208 (1992);  
  
    \bibitem{HS2} J.~Harnad and S.~Shnider,
  ``Isotropic geometry and twistors in higher dimensions II. Odd dimensions, reality conditions and twistor superspaces'', 
{\em J. Math. Phys.}, {\bf 36} 1945-1970 (1995). 

 \bibitem{JM}  M.~Jimbo and T.~Miwa. 
``Solitons and infinite-dimensional Lie algebras'', {\em  Publ. Res. Inst. Math. Sci.}, {\bf 19} 943-1001 (1983)

 \bibitem{Mac} I.~G.~Macdonald, {\em Symmetric Functions and Hall Polynomials},
Clarendon Press, Oxford, (1995).

 \bibitem{Oh} Y.~Ohta, ``Bilinear theory of solitons with Pfaffian labels'', 
{\em Kokyuroku, RIMS}, no. 822, 197-205 (1993).

\bibitem{Ok} S.~Okada, ``Pfaffian formulas and Schur Q-function identities'',
{\em Adv. Math.} {\bf 353} 446-470  (2019).

\bibitem{Sa} M.~Sato. ``Soliton equations as dynamical systems on infinite dimensional Grassmann manifold'' 
{\em Kokyuroku, RIMS} 30-46, (1981).

\bibitem{Shig1} Y.~Shigyo, ``On Addition Formulae of KP, mKP and BKP Hierarchies'',
SIGMA {\bf 9} 035 (2013).

\bibitem{Shig2} Y.~Shigyo, ``On the expansion coefficients of Tau-function
of the BKP hierarchy'',   {\em J. Phys. A} {\bf 49 } 295201 (2016).

 \bibitem{SW} G.~Segal  and G.~Wilson, ``Loop groups and equations of KdV type'', 
 {\em Publ. Math. IH\'ES} {\bf 6}, 5-65 (1985).
 
  \bibitem{VG} B. van Geemen, ``Schottky-Jung relations and vector bundles on hyperelliptic curves'',
{\em Math. Ann.} {\bf 281}, 431-450 (1988).

\bibitem{You} Y.~You,  ``Polynomial solutions of the BKP hierarchy and projective representations of symmetric groups'',
 in: {\em Infinite-Dimensional Lie Algebras and Groups},
 {\em Adv. Ser. Math. Phys.} {\bf 7}  (1989). World Sci. Publ., Teaneck, NJ.
 
 \bigskip
\noindent


\end{thebibliography}
\end{document}